\documentclass[preprint,12pt]{elsarticle}

\usepackage{amsmath,amsfonts,amsthm,amssymb,paralist,subfigure,graphicx,amsbsy,float,epsfig,color,xcolor}
\usepackage{cuted,mathtools,lipsum}

\usepackage[ruled,vlined,linesnumbered]{algorithm2e}

\usepackage{stmaryrd,url}

\usepackage{amsthm}

\newtheorem{lem}{Lemma}
\newtheorem{ass}{Assumption}
\newtheorem{theorem}{Theorem}

\newtheorem{rem}{Remark}
\newtheorem{prop}{Proposition}
\newtheorem{cor}{Corollary}
\newtheorem{prob}{Problem Formulation}

\def\mb{\mathbf}

\def\mc{\mathcal}
\def\wh{\widehat}
\def\wt{\widetilde}

\journal{Nonlinear Analysis: Hybrid Systems}

\begin{document}

\begin{frontmatter}

\title{Adding Heavy-Ball Momentum to Distributed Constrained Optimization for Energy Resource Management Subject to Time-Delays and Quantization 
}

\author[Sem]{Mohammadreza Doostmohammadian}
\affiliation[Sem]{Faculty of Mechanical Engineering, Semnan University, Semnan, Postal code: 35131-19111, Iran, Phone number: +9823-31530000, doost@semnan.ac.ir.}
\author[AA]{Sergio Pequito} 
\affiliation[AA]{ Department of Electrical and Computer Engineering, and Institute for Systems and Robotics, Instituto Superior Tecnico, University of Lisbon, Portugal, sergio.pequito@tecnico.ulisboa.pt}

\begin{abstract}
	{This paper studies a hybrid discrete-time signed gradient-consensus method with heavy-ball momentum, paired communication delays, and fixed logarithmic quantization. The method preserves the aggregate supply-demand equality at every iteration. For a static connected undirected graph and bounded delays, a Lyapunov\nobreakdash-Krasovskii analysis establishes practical stability by isolating the discontinuous quantization error as an additive disturbance and closing a computable invariant-sublevel argument. The resulting ergodic residual bound is quadratic in the quantization resolution, and exact convergence to the smooth penalized equality-constrained optimizer is recovered in the unquantized case. For delay\nobreakdash-free unquantized quadratic costs, an exact modal analysis identifies the heavy-ball stability region and optimal parameters, with the predicted rate validated on a fixed sparse network. A reproducible delayed experiment reports all constants and a certificate whose guaranteed steady-state level lies below the reported initialization. A separate mixed generator-battery experiment exercises the signed update directly. Finite smooth penalties approximate, but do not generally enforce, the local box constraints exactly.}
\end{abstract}

\begin{graphicalabstract}
 \includegraphics{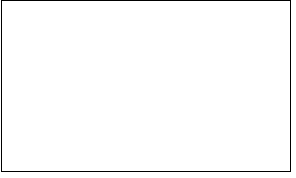}
\end{graphicalabstract}

\begin{highlights}
 \item Addressing optimal convergence and resource-demand feasibility for distributed generator-battery coordination
 \item Addressing communication delays for energy resource management.
 \item Addressing quantization as in real-world communication systems. 
\end{highlights}

\begin{keyword} 
	Distributed constrained optimization \sep energy resource scheduling \sep heavy-ball momentum \sep quantization \sep time-delays 
\end{keyword}

\end{frontmatter}

\section{Introduction} \label{sec_intro}

	Distributed energy resource management systems (DERMS) coordinate distributed energy resources, including solar generation, storage, electric vehicles, flexible loads, and small generators, to satisfy system-level operating objectives \cite{strezoski2023distributed,iweh2021distributed}. In the optimization model studied here, each resource is represented by a local controller that evaluates a private cost gradient and exchanges only a quantized signed-gradient value with neighboring controllers. The controllers collectively minimize the separable operating cost while maintaining the aggregate supply-demand equality. Thus, the DERMS interpretation is a networked economic-dispatch/resource-allocation layer rather than a complete model of distribution-system power flow, voltage regulation, or market clearing.
	
	Distributing the computation can improve scalability, reduce communication bottlenecks, and remove dependence on a single computational coordinator \cite{arc}. It does not, by itself, eliminate every single point of failure: communication cut sets, common infrastructure, and critical agents can still fail. Resilience therefore depends on the network topology and implementation. Parallel hardware and distributed scheduling can exploit resources more efficiently than centralized processing \cite{vidyarthi2008scheduling,EAAI}, supported by advances in convex/nonconvex decomposition \cite{10035454}, consensus methods \cite{garin2010survey}, and asynchronous protocols \cite{jiang2023distributed}.

By distributing control and coordination across many agents (e.g., local controllers, microgrids), DERMS can adapt faster to local demand conditions, optimize resource use with lower latency, and maintain the supply-demand balance more effectively. Decentralized approaches also reduce coordination costs, enable privacy-preserving transactions \cite{zholbaryssov2020privacy}, and allow parallel computation and data processing \cite{zhang2020distributed,li2021parallel} without relying on a central infrastructure. These attributes make DERMS particularly relevant for modern power systems, edge computing applications, and networked resource environments in smart cities \cite{tahmasebi2024security}.

\subsection{Related Work}

	For broad overviews of distributed optimization and its control applications, we refer to \cite{yang2019survey,nedic2018distributed}. These surveys are distinguished below from application papers and individual algorithms.

The literature on distributed optimization for DERMS has evolved from early linear and ideal multi-agent models \cite{amjady2009economic,kar2014distributed,nunna2012multiagent,kar2012distributed,golmohamadi2019multi} to more advanced formulations that account for practical non-idealities.

\textit{DERMS optimization architectures.}
Game-theoretic approaches for microgrid energy management via graph analysis are studied in \cite{10599894}. Decentralized blockchain strategies for distributed energy resources in virtual power plants are developed in \cite{yang2021blockchain}. Machine learning and big-data methods for integrating renewable resources into the grid are investigated in \cite{mostafa2022renewable}, with extensions to hybrid energy sources and battery storage in \cite{kumar2024multi}.

\textit{Distributed constrained optimization for economic dispatch.}
A focused line of research formulates DERMS as distributed constrained optimization, as surveyed in \cite{arc}. Foundational gradient-based methods with anytime feasibility are proposed in \cite{lcss,cherukuri2015distributed,cherukuri2016initialization}. Extensions include adaptive optimization over unbalanced networks \cite{lian2023distributed}, ADMM-based solutions \cite{zhu2025deed,10740459,liu2023distributed0,han2024distributed}, saturated dynamics for power ramp-rate limits \cite{dsp}, multi-step gradient descent for economic dispatch \cite{bai2022weighted}, and consensus-based multi-agent power flow analysis \cite{saric2024distributed}.

\textit{Accelerated distributed optimization.}
The convergence rate of distributed algorithms is a critical concern, as faster convergence translates to fewer communication rounds, lower computation overhead, and reduced risk of oscillations when system conditions change. Acceleration via heavy-ball momentum has been studied in the general distributed optimization literature \cite{xin2019distributed,jin2022momentum}, along with Nesterov-type momentum \cite{zheng2023distributed} and non-Lipschitz sign-based dynamics for fixed-time convergence \cite{spl,liu2023distributed}. Improved convergence for DERMS specifically is also investigated in \cite{shan2024distributed}. Push-sum and push-pull methods \cite{nedic2015distributed,pu2021distributed} address directed and time-varying topologies but typically do not consider quantization or coupling constraints jointly.

\textit{Communication impairments.}
Data quantization due to limited channel bit-rate is a significant practical concern. Uniform quantization for distributed optimization is studied in \cite{michelusi2022finite,10748356,9521982}, and logarithmic quantization in \cite{TASE,kalantari2010logarithmic,lcss2}. Time-delays in distributed algorithms are addressed in several works (e.g., \cite{lcss}), but typically not jointly with quantization and momentum acceleration.

	Two recent delay-robust power-system results clarify the scope of the present contribution. Shi \emph{et al.} \cite{shi2025delayed} treat economic dispatch with heterogeneous delays, bounded generation, adaptive penalties, and a continuous-time primal-dual/scattering construction. Li \emph{et al.} \cite{li2023delay} establish delay-independent stability for distributed secondary frequency control through a reformulated continuous-time primal-dual controller. In contrast, the present method is a discrete-time signed-gradient allocation scheme with exact preservation of the aggregate equality under paired symmetric delays, fixed logarithmic quantization, and heavy-ball momentum. Its delay guarantee is delay-dependent and assumes a static undirected graph; it neither subsumes the scattering-based heterogeneous-delay result nor provides frequency-dynamics guarantees.

\subsection{Gap and Contributions}

{The main technical contribution is a distributed momentum\nobreakdash-based signed gradient-consensus method that preserves the coupling equality at every iteration. In particular, paired delays and fixed logarithmic quantization coexist with exact per\nobreakdash-iteration conservation of the coupling equality. The delayed fixed-quantized recursion is naturally a hybrid dynamical object: continuous resource states evolve through discrete quantizer-cell and delay selections. The practical\nobreakdash-stability theorem is stated for a static connected graph; the feasibility identity itself also holds for paired time-varying undirected edges.} The specific results are as follows:
\begin{enumerate}[(i)]
	\item \textit{Ergodic practical-stability under quantization, momentum, and delays:} {We prove a self-contained Lyapunov\nobreakdash-Krasovskii result. The delay mismatch is decomposed into a smooth-gradient history term and an additive discontinuous-quantization term, and a computable invariant-sublevel condition removes the former a priori orbit assumption. This additive-disturbance treatment avoids any Lipschitz assumption on the discontinuous quantizer and is reusable in delayed quantized-network analyses.} The resulting average residual is explicit and scales as $O(\rho^2)$ (Theorem~\ref{thm_conv}). The exact-convergence result is recovered in the unquantized case $\rho = 0$ (Corollary~\ref{cor_exact_rho0}).
	
	\item \textit{Anytime feasibility under symmetric delays:} We prove that the supply-demand coupling constraint is satisfied at every iteration, even in the presence of log-scale quantization, momentum, and bounded symmetric delays (Lemma~\ref{lem_feas}). 
	
	\item \textit{Quantized equilibria and approximate optimality:} We characterize the equilibria of the fixed-quantized dynamics as quantized-consensus points and show that every such equilibrium is feasible and satisfies an explicit $O(\rho^2)$ optimality-gap bound with respect to the penalized problem (Lemma~\ref{lem_quantized_equilibrium} and Proposition~\ref{prop_gap_quantized_eq}). The bound is stated using the exact logarithmic-quantizer constant and, therefore, remains valid for the quantizer defined.
	
	\item \textit{Spectral momentum characterization:} {For the delay-free unquantized quadratic case, we derive the exact modal stability condition, the optimal heavy-ball pair $(\eta_{\rm HB},\mu_{\rm HB})$, and its worst-case contraction factor (Proposition~\ref{prop_spectral_hb}).}
\end{enumerate}

\subsection{Paper Organization}

Section~\ref{sec_fram} formulates the DERMS optimization problem and states the standing assumptions. Section~\ref{sec_dyn} develops the momentum-based distributed algorithm and its extension to delayed communication. Section~\ref{sec_analysis} provides the feasibility, optimality, and convergence analysis. Section~\ref{sec_sim} presents numerical experiments, and Section~\ref{sec_conc} concludes the paper.

\textit{Notation.} Bold symbols denote vectors and matrices; scalar components and dummy arguments are not bold. The vector $\mb{1}_n$ denotes the all-ones vector of size~$n$. We define the \textit{signed vector} $\mb{s}_n := [\mb{1}_{n_1}; -\mb{1}_{n_2}] \in \mathbb{R}^n$ and the signing matrix $B := \mathrm{diag}(b_i)$ with $b_i = 1$ for generators and $b_i = -1$ for batteries. The operator ``$;$'' denotes column concatenation. We write $\nabla_i$ for $\frac{d f_i}{d z_i}$ and $\nabla_z F$ for the gradient of $F(\cdot)$ with respect to~$\mb{z}$.

\section{Problem Formulation} \label{sec_fram}

\subsection{Constrained Optimization Formulation}

The goal of the energy resource management problem is to minimize the total energy cost across $n_1$ generator nodes (producing power) and $n_2$ battery nodes (storing power). This is formulated as \cite{mahmoodi2015economic}:
\begin{align} \label{eq_dra}
	\min_{\mb{x},\mb{y}}
	~~ & F(\mb{x},\mb{y}) = \sum_{i=1}^{n_1} g_i(x_i) + \sum_{i=1}^{n_2} h_i(y_i),
\end{align}
where $x_i \in \mathbb{R}$ denotes the power output of generator $i \in \{1,\dots,n_1\}$, $y_i \in \mathbb{R}$ denotes the power assigned to battery $i \in \{1,\dots,n_2\}$, and $g_i:\mathbb{R} \to \mathbb{R}$, $h_i:\mathbb{R} \to \mathbb{R}$ are the local generation and storage cost functions, respectively.

For notational convenience, let $n := n_1 + n_2$ and define $\mb{z} = [\mb{x};\mb{y}] \in \mathbb{R}^n$. Further, let $f_i(z_i) := g_i(z_i)$ for $i \in \{1,\dots,n_1\}$ (generators) and $f_i(z_i) := h_{i-n_1}(z_i)$ for $i \in \{n_1+1,\dots,n\}$ (batteries). The problem~\eqref{eq_dra} takes the compact form:
\begin{align} \label{eq_dra2}
	\min_{\mb{z}}
	~~ & F(\mb{z}) = \sum_{i=1}^{n} f_i(z_i).
\end{align}

The optimization~\eqref{eq_dra2} is subject to a coupling equality constraint representing the supply-demand balance \cite{mahmoodi2015economic}:
\begin{align} \label{eq_b}
	\sum_{i=1}^{n_1} x_i  = D + \sum_{i=1}^{n_2} y_i,
\end{align}
where $D \in \mathbb{R}^+$ denotes the total power demand. Equivalently, given initial allocations $d^x_i$ and $d^y_i$ satisfying $\sum_{i=1}^{n_1} d^x_i - \sum_{i=1}^{n_2} d^y_i = D$, the constraint can be written as:
\begin{align} \label{eq_b2}
	\sum_{i=1}^{n_1} (x_i - d^x_i)  =  \sum_{i=1}^{n_2} (y_i - d^y_i).
\end{align}
A key requirement is \textit{anytime resource-demand feasibility}: the coupling constraint~\eqref{eq_b2} must hold at every iteration $k \in \mathbb{N}$ of the algorithm.

In addition, each energy node is subject to local box constraints reflecting physical limits:
\begin{align} \label{eq_box}
	m_i \leq z_i(k) \leq M_i, \quad \forall\, i \in \{1,\dots,n\},~\forall\, k,
\end{align}
where $m_i$ and $M_i$ denote the minimum and maximum power capacity of node~$i$.

\subsection{Standing Assumptions}

\begin{ass}[Strong convexity and Lipschitz gradient] \label{ass_cvx}
	Each local \emph{composite penalized} cost function $f_i:\mathbb{R} \to \mathbb{R}$ used by the dynamics is twice continuously differentiable and globally strongly convex with globally Lipschitz-continuous gradient. Specifically, there exist constants $0 < v \leq u$ such that
		\begin{align} \label{eq_vu}
			2v \leq \partial^2_{z_i} f_i(z_i) \leq 2u, \quad \forall\, z_i \in \mathbb{R},~\forall\, i.
		\end{align}
		The global domain is required because a finite smooth penalty does not project an iterate onto $[m_i,M_i]$.
\end{ass}

\begin{rem}
	{Under the convention~\eqref{eq_vu}, the Hessian strong-convexity modulus is $2v$, while the quadratic lower bound is $F(\mb z)-F(\mb z^\star)\geq v\|\mb z-\mb z^\star\|^2$ on the feasible affine space.} For the quadratic costs $f_i(z_i) = \gamma_i z_i^2 + \beta_i z_i + \alpha_i$ with $\gamma_i > 0$, Assumption~\ref{ass_cvx} holds with $v = \min_i \gamma_i$ and $u = \max_i \gamma_i$.
\end{rem}

\subsection{Log-Scale Quantization}

In practical DERMS, data exchanged over communication channels is subject to quantization due to limited bit-rate. This paper adopts log-scale quantization \cite{TASE,kalantari2010logarithmic,lcss2}, which yields a significantly smaller optimality gap near convergence compared to uniform quantization \cite{michelusi2022finite,10748356,9521982}. To make the unquantized limit well defined, we distinguish it from the logarithmic formula:
	\begin{align}
		q_0(s)&:=s, \nonumber\\
		q_\rho(0)&:=0,\qquad \rho>0,\nonumber\\
		q_\rho(s)&:=\mbox{sgn}(s)\exp\left(\rho \left[\dfrac{\log(|s|)}{\rho}\right] \right),
		\quad \rho>0,~s\ne0. \label{eq_quan_log}
	\end{align}

\noindent where $[\cdot]$ denotes rounding to the nearest integer and $\rho$ is the quantization resolution parameter. For $\rho>0$, this mapping is odd ($q_\rho(-s) = -q_\rho(s)$) and, by the rounding-to-nearest construction, satisfies the exact multiplicative bound
\begin{align}
	e^{-\rho/2} \leq \frac{q(s)}{s} \leq e^{\rho/2}, \quad \forall\, s \neq 0. \label{eq_quan_sect}
\end{align}
Define
\begin{align}
	\delta_\rho := e^{\rho/2} - 1. \label{eq_delta_rho_def}
\end{align}
Then~\eqref{eq_quan_sect} implies the componentwise relative-error bound
\begin{align}
	|q(s)-s| \leq \delta_\rho |s|, \quad \forall\, s \in \mathbb{R}. \label{eq_quan_relerr}
\end{align}
As illustrated in Fig.~\ref{fig_quant}, logarithmic quantization is \textit{strongly sign-preserving}: $\frac{q(s)}{s}$ remains bounded away from zero for all $s \neq 0$. By contrast, uniform quantization has a dead-zone around the origin that prevents exact convergence (see Section~\ref{sec_sim}).

\begin{figure}[!h]
	\centering
	\includegraphics[width=3.95in]{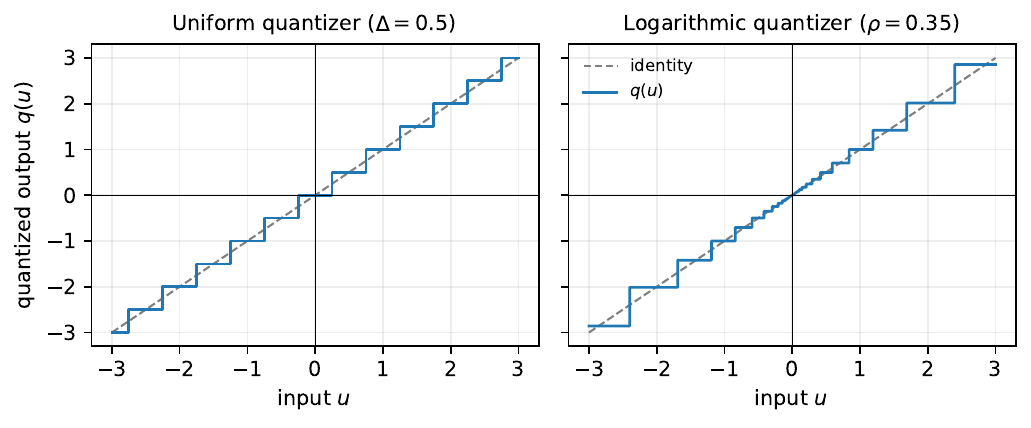}
	\caption{{Uniform quantization (left) and logarithmic quantization (right).}}
	\label{fig_quant}
\end{figure}

\begin{rem} \label{rem_sector}
	The subsequent analysis uses the odd symmetry and componentwise relative-error property of the quantizer. In particular, with $e_q(s) := q_\rho(s) - s$, \eqref{eq_quan_relerr} gives $|e_q(s)| \leq \delta_\rho |s|$ for all $s \in \mathbb{R}$, where $\delta_\rho = e^{\rho/2}-1 = O(\rho)$. {The proof never treats $q_\rho$ as a Lipschitz mapping: its discontinuous error is isolated as an additive disturbance.} For the fixed-quantized dynamics, this yields residual terms of order $O(\delta_\rho^2)=O(\rho^2)$ in the Lyapunov descent inequality; exact convergence is recovered in the special case $\rho = 0$.
\end{rem}

\begin{rem}[Alphabet size and finite-bit implementation] \label{rem_bits}
		The ideal logarithmic quantizer in~\eqref{eq_quan_log} has countably many nonzero levels and is the object covered by the theorem. A practical implementation may retain indices $\ell_{\min},\ldots,\ell_{\max}$. With $M_q=\ell_{\max}-\ell_{\min}+1$, the signed levels plus zero has $2M_q+1$ symbols and requires $\lceil\log_2(2M_q+1)\rceil$ bits per scalar (apart from protocol overhead). The relative-error analysis remains valid while all transmitted magnitudes lie within the retained range. Underflow or saturation adds a separate bounded error and is not covered by Theorem~\ref{thm_conv}; the selected dynamic range must therefore be reported in a finite-bit deployment.
\end{rem}

\subsection{Communication Network}

	The communication network among energy nodes is modeled as a static undirected graph $\mc{G}=(\mc V,\mc E)$. Let $A=A^\top\geq0$ be its weighted adjacency matrix, $\mc N_i:=\{j:A_{ij}>0\}$, $\mc D:=\operatorname{diag}(\sum_jA_{ij})$ its degree matrix, and $L=\mc D-A$ its Laplacian. We denote the smallest positive and largest eigenvalues of $L$ by $\lambda_2$ and $\lambda_n$, and the maximum weighted degree by $d_{\max}:=\max_i\sum_jA_{ij}$. The static assumption is part of the corrected practical-stability theorem.

\begin{ass}[Connectivity] \label{ass_conn}
	{The graph $\mc G$ is connected, i.e., $\lambda_2>0$.}
\end{ass}

\begin{rem}
		This assumption is stronger than $B$-connectivity. Under mere $B$-connectivity an instantaneous Laplacian can have $\lambda_2(L(k))=0$, invalidating the one-step coercivity used below. A changing graph also requires a delay model that specifies messages on appearing and disappearing edges. Therefore, no convergence claim under $B$-connectivity is made; it would require a windowed Lyapunov analysis.
\end{rem}

\textit{Time-delays:}
The links in $\mc{G}$ may be subject to communication delays. Let $\tau_{ij}(k) \geq 0$ denote the delay at link $(j,i)$ at iteration~$k$.

\begin{ass}[Bounded symmetric delays] \label{ass_delay}
	The delays satisfy: (i)~$\tau_{ij}(k) \leq \bar\tau < \infty$ for all links $(j,i)$ and all $k$; and (ii)~$\tau_{ij}(k) = \tau_{ji}(k)$ for all active links and all~$k$. The delays are otherwise arbitrary (possibly time-varying and heterogeneous across links).
\end{ass}

\begin{rem}
	The symmetry condition $\tau_{ij}(k) = \tau_{ji}(k)$ ensures that for every undirected edge $\{i,j\}$, the two endpoint updates use the same delayed pair at each iteration. This is essential for the pairwise cancellation argument in the feasibility proof (Lemma~\ref{lem_feas}). {It is an implementation requirement rather than a property of arbitrary wireless delays: a time-stamped edge transaction, handshake, or paired-release buffer can make both endpoints apply the same delay index. Node $i$ must know the realized per-edge delay and retain its own past quantized gradients so that $q(\nabla_i(k-r))$ is paired with the received $q(\nabla_j(k-r))$.} Fully directional asynchronous delays would require a redesigned conservation mechanism, which we leave for future work.
\end{rem}

\subsection{Problem Statement}

\begin{prob} \label{prob_main}
	{Compute a distributed solution of the \emph{smooth penalized approximation} of~\eqref{eq_dra2} on the coupling hyperplane~\eqref{eq_b2}, under Assumptions~\ref{ass_cvx}-\ref{ass_delay} and with log-quantized channels~\eqref{eq_quan_log}. The original box-constrained problem is recovered exactly only under the exact-penalty conditions of Proposition~\ref{prop_penalty}, or approached as the smooth penalty is sharpened.} The solution must satisfy:
	\begin{enumerate}[(i)]
		\item \textit{Anytime equality feasibility:} the coupling constraint~\eqref{eq_b2} holds at every iteration; {finite penalties do not imply per-iterate box feasibility};
		\item \textit{Practical convergence under fixed quantization:} under bounded symmetric delays and connected graphs, the iterates satisfy an explicit ergodic residual bound for the penalized problem; exact convergence is recovered in the unquantized case $\rho=0$;
		\item \textit{Acceleration:} convergence is improved via heavy-ball momentum.
	\end{enumerate}
\end{prob}

\section{Momentum-Based Distributed Optimization} \label{sec_dyn}

This section develops the proposed distributed algorithm for solving Problem~\ref{prob_main}.

\subsection{Handling Box Constraints via Penalty Functions} \label{sec_penalty}

	To approximate the box constraints~\eqref{eq_box} within the distributed dynamics, we adopt a penalty-based approach \cite{bertsekas1975necessary,wu2021new,10925624}. Each original local cost $f_i^0(z_i)$ is augmented with a penalty term $f_i^\sigma(z_i)$, and the algorithm optimizes the \emph{penalized} composite function $\wh{f}_i(z_i) := f_i^0(z_i) + f_i^\sigma(z_i)$. A finite smooth penalty does not guarantee $z_i(k)\in[m_i,M_i]$ at every iteration; ``anytime feasibility'' in this paper refers only to the coupling equality~\eqref{eq_b2}.

Two choices for the penalty term are
\begin{align} \label{eq_log_penalty0}
	f_i^\sigma(z_i) &:= \sigma [z_i - M_i]^+ + \sigma [m_i - z_i]^+,
\end{align}
{with $[u]^+ := \max \{u,0\}$ and $\sigma>0$; and the smooth softplus approximation}
\begin{align} \label{eq_log_penalty}
	f_i^\sigma(z_i) &:= \frac{\sigma}{\kappa} \log (1+e^{\kappa (z_i - M_i)}) + \frac{\sigma}{\kappa} \log (1+e^{\kappa (m_i - z_i)}),
\end{align}
with $\kappa>0$.

\begin{rem}[Effect of penalty on problem constants] \label{rem_penalty}
	Adding the penalty modifies the smoothness constant. If the original Hessian is bounded above by $2u_0$, then the two softplus terms contribute
		\[
		\sigma\kappa\{s_+(1-s_+)+s_-(1-s_-)\}\leq \sigma\kappa/2,
		\]
		where $s_+=\bigl(1+e^{-\kappa(z_i-M_i)}\bigr)^{-1}$ and $s_-=\bigl(1+e^{-\kappa(m_i-z_i)}\bigr)^{-1}$. Thus the conservative global choice $\wh u=u_0+\sigma\kappa/4$ is valid. {Writing $W_i=M_i-m_i$ and splitting the real line at the box midpoint gives the sharper global bound
			\begin{align}
				\wh u\leq u_0+\frac{\sigma\kappa}{8}
				+\frac{\sigma\kappa}{2}e^{-\kappa W_i/2}, \label{eq_penalty_u_sharp}
			\end{align}
			because one logistic curvature factor is at most $1/4$ and the other at most $e^{-\kappa W_i/2}$.} Every theorem coefficient must use a valid \emph{composite} $\wh u$, not the curvature of the unpenalized quadratic alone. For notational simplicity, $f_i,F,u$ below denote the smooth composite cost and its constants.
\end{rem}

\begin{prop}[Exact hinge penalty and smooth approximation]
	\label{prop_penalty}
		Let $\mc Z_D:=\{\mb z:\mb s_n^\top\mb z=D\}$, let $\mb z^\star_{\rm box}$ be the unique optimizer of the original cost $F_0$ over $\mc Z_D\cap[m,M]$, and suppose KKT multipliers $\lambda_i^+,\lambda_i^-\geq0$ exist for $z_i\leq M_i$ and $z_i\geq m_i$. Define
		\[
		\sigma_{\rm ex}:=\max_i\{\lambda_i^+,\lambda_i^-\}.
		\]
		For every $\sigma>\sigma_{\rm ex}$, the hinge-penalized optimizer $\mb z_\sigma^\star$ over $\mc Z_D$ equals $\mb z^\star_{\rm box}$. For the softplus-penalized optimizer $\mb z_{\sigma,\kappa}^\star$ and strong-convexity constant $v$,
		\begin{align}
			\|\mb z_{\sigma,\kappa}^\star-\mb z_\sigma^\star\|
			\leq \sqrt{\frac{2n\sigma\log 2}{v\kappa}}, \qquad
			\mb z_{\sigma,\kappa}^\star\to \mb z_\sigma^\star~~(\kappa\to\infty). \label{eq_penalty_gap}
		\end{align}
		Consequently, a fixed finite $(\sigma,\kappa)$ generally solves only a smooth approximation; it need not satisfy the box constraints exactly.
\end{prop}

\begin{proof}
		The first claim is the convex exact $\ell_1$-penalty theorem: the KKT inequality multipliers give the threshold above which every violation costs more than its possible first-order decrease in $F_0$ \cite{bertsekas1975necessary}. For the smooth approximation, $0\leq \kappa^{-1}\log(1+e^{\kappa t})-[t]^+\leq(\log2)/\kappa$. Thus the softplus objective lies \emph{one-sidedly} between the hinge objective and the hinge objective plus $\Delta:=2n\sigma\log2/\kappa$. Since the Hessian convention in Assumption~\ref{ass_cvx} is $\nabla^2F\succeq2vI$, optimality of the two minimizers gives $v\|\mb z_{\sigma,\kappa}^\star-\mb z_\sigma^\star\|^2\leq\Delta$, which proves~\eqref{eq_penalty_gap}. {The one-sided comparison is why no additional factor of two appears.}
\end{proof}

\begin{rem}[Scope of the subsequent convergence claims]
	\label{rem_scope_claims}
	Proposition~\ref{prop_penalty} concerns only the relation between the original box-constrained problem and its penalized counterparts.
	The dynamical results proved later in Theorem~\ref{thm_conv} and Corollary~\ref{cor_exact_rho0} apply to the penalized distributed dynamics.
	In particular, for fixed quantization level \(\rho>0\), the analysis yields practical stability and an \(O(\rho^2)\) residual neighborhood around the penalized optimizer.
	Exact asymptotic convergence is recovered in the unquantized case \(\rho=0\).
	{The theorem requires the smooth penalty~\eqref{eq_log_penalty}; the nonsmooth hinge~\eqref{eq_log_penalty0} is outside Assumption~\ref{ass_cvx}. Relation to the box-constrained optimizer is obtained separately through Proposition~\ref{prop_penalty}, and no per-iterate box-feasibility claim is made.}
\end{rem}

\subsection{Delay-Free Dynamics} \label{sec_delayfree}

We present the algorithm first in the delay-free case ($\bar\tau = 0$). The dynamics is based on a \textit{signed gradient-consensus} mechanism: each node drives its signed local gradient $b_i \nabla_i$ toward agreement with its neighbors, which at equilibrium yields the optimality condition (Lemma~\ref{lem_z*}). 

Each energy node $i$ updates {(we, henceforth, write $q=q_\rho$)}:
\begin{align} \label{eq_sol_nodelay}
	z_i(k+1) &= z_i(k) + \eta b_i \sum_{j \in \mc{N}_i} A_{ij} \Big( b_j q(\nabla_j(k)) - b_i q(\nabla_i(k)) \Big) \\
	&\qquad\qquad+ \mu w_i(k), \notag \\
	\label{eq_w_nodelay}
	w_i(k+1) &= z_i(k+1) - z_i(k),
\end{align}
where $\eta > 0$ is the step size and $0 < \mu < 1$ is the heavy-ball momentum coefficient. Setting $\mu = 0$ recovers the non-accelerated dynamics.

\textit{Compact form:} Define the \textit{signed gradient} $\boldsymbol{\phi}(k) := B \nabla_z F(\mb{z}(k))$ and its quantized version $\mb{g}(k) := B\, q(\nabla_z F(\mb{z}(k)))$, where quantization is applied componentwise. With $\mu = 0$, the dynamics~\eqref{eq_sol_nodelay} reads in vector form:
\begin{align} \label{eq_L}
	\mb{z}(k+1) &= \mb{z}(k) - \eta\, B\, L\, \mb{g}(k)\notag \\
	&= \mb{z}(k) - \eta\, B\, L\, B\, q(\nabla_z F(\mb{z}(k))).
\end{align}
Note the two occurrences of $B$: the outer $B$ arises from the $b_i$ prefactor in~\eqref{eq_sol_nodelay}, and the inner $B$ converts the quantized gradient $q(\nabla_z F)$ into the signed quantity~$\mb{g}$.

\textit{Initialization:} For anytime feasibility, the initial states must satisfy:
\begin{align} \label{eq_init}
	z_i(0) = d_i, \quad w_i(0) = 0, \quad \forall\, i,
\end{align}
where $d_i := d^x_i$ for generators and $d_i := d^y_i$ for batteries, with $\sum_{i=1}^{n_1} d^x_i - \sum_{i=1}^{n_2} d^y_i = D$.

\begin{prop}[Spectral heavy-ball characterization for quadratic costs] \label{prop_spectral_hb}
		Consider the static, delay-free, unquantized case and a quadratic composite cost
		$F(z)=\frac12z^\top Qz+r^\top z$, where $Q=\operatorname{diag}(h_i)\succ0$.
		Let $0<\theta_{\min}\leq\theta_{\max}$ be the smallest and largest positive eigenvalues of $Q^{1/2}LQ^{1/2}$. On the equality\nobreakdash-feasible subspace, every modal error satisfies
		\begin{align}
			e_\theta(k+1)=(1+\mu-\eta\theta)e_\theta(k)-\mu e_\theta(k-1). \label{eq_hb_mode}
		\end{align}
		Thus all modes are Schur stable if and only if
		\begin{align}
			0\leq\mu<1,\qquad 0<\eta\theta_{\max}<2(1+\mu). \label{eq_hb_schur}
		\end{align}
		The minimax heavy-ball choice and contraction factor are
		\begin{align}
			\eta_{\rm HB}&=\frac{4}{(\sqrt{\theta_{\max}}+\sqrt{\theta_{\min}})^2},\nonumber\\
			\mu_{\rm HB}&=\left(\frac{\sqrt{\theta_{\max}}-\sqrt{\theta_{\min}}}{\sqrt{\theta_{\max}}+\sqrt{\theta_{\min}}}\right)^2, \label{eq_hb_opt}\\
			r_{\rm HB}&=\frac{\sqrt{\varkappa}-1}{\sqrt{\varkappa}+1},\qquad \varkappa:=\theta_{\max}/\theta_{\min}. \label{eq_hb_rate}
		\end{align}
		For comparison, optimally tuned gradient descent has factor $r_{\rm GD}=(\varkappa-1)/(\varkappa+1)$. Hence, $r_{\rm HB}<r_{\rm GD}$ whenever $\varkappa>1$.
\end{prop}
\begin{proof}
		Let $\mb x(k)=\mb z(k)-\mb z^\star$ and $\mb y(k)=B\mb x(k)$. Equality feasibility gives $\mathbf1^\top\mb y(k)=0$, and the update becomes $\mb y(k+1)=(I-\eta LQ)\mb y(k)+\mu(\mb y(k)-\mb y(k-1))$. The nonzero spectrum of $LQ$ equals that of the symmetric matrix $Q^{1/2}LQ^{1/2}$, yielding~\eqref{eq_hb_mode}. The second-order Jury conditions give~\eqref{eq_hb_schur}; minimizing the maximum root modulus over $[\theta_{\min},\theta_{\max}]$ gives~\eqref{eq_hb_opt}-\eqref{eq_hb_rate} \cite{polyak1964some}.
\end{proof}

\begin{rem}
	{Proposition~\ref{prop_spectral_hb} provides a principled momentum choice for the quadratic baseline, and Section~\ref{sec_spectral_validation} validates its predicted modal recurrence on an actual fixed sparse network. At the minimax pair, the two extreme modes have repeated roots; hence $r_{\rm HB}$ is the asymptotic root-modulus factor, while a finite-time bound may include a polynomial prefactor. For nonlinear, delayed, or quantized dynamics, Theorem~\ref{thm_conv} gives a sufficient robust region rather than an optimal rate; the general $\mu^*(\eta,\bar\tau,\rho)$ remains open.}
\end{rem}

\subsection{Extension to Delayed Communication} \label{sec_delaycase}

Under Assumption~\ref{ass_delay}, we introduce the indicator function:
\begin{align} \label{eq_I}
	\mathcal{I}_{k,ij}(r) := \left\{ \begin{array}{ll}
		1, & \text{if}~  \tau_{ij}(k) = r,\\
		0, & \text{otherwise}.
	\end{array}\right.
\end{align}
The delayed dynamics is given by
\begin{align} \nonumber
	z_i(k+1) ={}& z_i(k) + \eta_\tau b_i \sum_{j \in \mc{N}_i}  \sum_{r=0}^{\bar\tau} A_{ij}  \Big(b_j q(\nabla_j(k-r)) \\ &- b_i q(\nabla_i(k-r))\Big) {\mathcal{I}_{k,ij}(r)} + \mu w_i(k), \label{eq_sol_delay1}\\
	w_i(k+1) ={}& z_i(k+1) - z_i(k), \label{eq_sol_delay2}
\end{align}
where $\eta_\tau > 0$ is the step size for the delayed case. Setting $\bar\tau = 0$ recovers~\eqref{eq_sol_nodelay}-\eqref{eq_w_nodelay} with $\eta_\tau = \eta$.
{For $k<0$, set $z(k)=z(0)$ and $w(k)=0$. At each update time $k$, exactly one indicator $\mathcal I_{k,ij}(r)$ is active per edge. This convention is used in the delay-mismatch bound below.}

Algorithm~\ref{alg_1} summarizes the procedure.
\begin{algorithm}
	\textbf{Input:}  $A$, $\mc{G}$, $\eta_\tau$, $\mu$, $\sigma$, $\{d_i\}_{i=1}^n$, $\{f_i(\cdot)\}_{i=1}^n$, $\overline{\tau}$\;
	\textbf{Initialization:} $k=0$, $z_i(0) = d_i$, $w_i(0) = 0$ for all $i$\; allocate a local gradient-history buffer of length $\bar\tau+1$\;
	\While{not converged}{
		\ForEach{energy node $i \in \mc{V}$ \textbf{in parallel}}{
			compute and store $q(\nabla_i(k))$ and transmit its index\;
			for each edge $\{i,j\}$ obtain the common realized delay $r=\tau_{ij}(k)=\tau_{ji}(k)$ and retrieve the local value $q(\nabla_i(k-r))$ paired with the received $q(\nabla_j(k-r))$\;
			update $z_i(k+1)$ by~\eqref{eq_sol_delay1} and $w_i(k+1)$ by~\eqref{eq_sol_delay2}\;
		}
		$k \leftarrow k+1$\;
	}
	\textbf{Return:} $\{z_i(k)\}_{i=1}^n$ and $\wh{F}(\mb{z}(k))$\;
	\caption{Heavy-ball momentum DERMS dynamics}
	\label{alg_1}
\end{algorithm}

\subsection{Computational, Communication, and Storage Complexity}

We briefly quantify the per-iteration arithmetic, communication, and memory requirements of Algorithm~\ref{alg_1}.
Since each agent maintains a scalar decision variable \(z_i\) and the local gradients are available in closed form, the local gradient and quantizer evaluations are constant-time operations.
The dominant per-iteration cost therefore comes from aggregating information over neighboring agents.

\begin{prop}[Per-iteration complexity of Algorithm~\ref{alg_1}]
	\label{prop_complexity}
	Consider Algorithm~\ref{alg_1} over $\mc G=(\mc V,\mc E)$ with neighbor set \(\mathcal N_i\) for agent \(i\), and let \(\bar\tau\) denote the maximum communication delay. Suppose a finite implementation retains $M_q$ positive logarithmic levels as in Remark~\ref{rem_bits}.
	
	For scalar local decision variables, one iteration of Algorithm~\ref{alg_1} at node \(i\) requires:
	\begin{enumerate}
		\item \(O(1)\) arithmetic operations to evaluate the local gradient and its quantized version;
		\item \(O(|\mathcal N_i|)\) arithmetic operations to aggregate neighbor contributions in the update;
		\item transmission and reception of one quantizer index per active neighbor, using $\lceil\log_2(2M_q+1)\rceil$ bits per index apart from protocol metadata;
		\item {$O(\bar\tau\,(|\mathcal N_i|+1))$ scalar memory} to store delayed neighbor values and the node's own gradient history.
	\end{enumerate}
	
	Hence, the local per-iteration arithmetic cost is \(O(|\mathcal N_i|)\), while the network-wide arithmetic cost is
	\[
	O\!\left(\sum_{i=1}^n |\mathcal N_i|\right)=O(|\mc E|).
	\]
	Likewise, the network-wide communication cost per iteration is
	{$2|\mc E|\lceil\log_2(2M_q+1)\rceil$ bits} apart from protocol metadata, and the total delay-buffer storage cost is $O(\bar\tau(|\mc E|+n))$.
	
	In particular, if the graph has bounded degree, i.e., \( |\mathcal N_i| \le \Delta_{\max}\) for all \(i\), then the total per-iteration arithmetic and communication costs scale linearly with \(n\). For dense graphs, the corresponding costs scale as \(O(n^2)\).
\end{prop}

\begin{proof}
	At each iteration, agent \(i\) performs one local gradient evaluation, one quantizer evaluation, and one update of its own primal/momentum variables; these are \(O(1)\) operations under the scalar-node model adopted in this paper.
	The remaining work is the neighbor aggregation appearing in the distributed update, which requires one contribution per active neighbor, yielding \(O(|\mathcal N_i|)\) arithmetic operations locally.
	
	For communication, each edge $\{i,j\}\in\mc E$ carries one quantizer index in each direction per iteration, giving the stated bit count from Remark~\ref{rem_bits}.
	
	For storage, bounded delays require each agent to retain delayed neighbor information and its own quantized-gradient history over at most \(\bar\tau\) past steps.
	Thus, node \(i\) stores at most $O(\bar\tau(|\mathcal N_i|+1))$ scalar values.
	Summing the local arithmetic and storage costs over all agents and using
	\[
	\sum_{i=1}^n |\mathcal N_i| = 2|\mc E|,
	\]
	gives the stated network-wide bounds.
	The bounded-degree and dense-graph scalings follow immediately.
\end{proof}

\begin{rem}
	\label{rem_complexity_vector}
	The above statement is given for the scalar-agent setting considered in this paper.
	If each agent instead controls a vector in \(\mathbb R^p\), then all arithmetic, communication, and storage bounds scale linearly in \(p\).
\end{rem}

\section{Analysis of Feasibility, Optimality, and Convergence} \label{sec_analysis}

The analysis proceeds in three layers. First, we study the distributed dynamics for the smooth penalized equality-constrained problem and establish anytime feasibility together with practical stability of the fixed-quantized algorithm. Second, we characterize the equilibria of the fixed-quantized dynamics and bound their suboptimality relative to the penalized optimizer. Third, in the special case $\rho=0$, the residual vanishes and exact convergence to the penalized optimizer is recovered. The approximation of the original box-constrained problem by the penalized problem is handled separately through Proposition~\ref{prop_penalty}.

\subsection{Optimality Condition}

\begin{lem}[Optimality condition] \label{lem_z*}
	\cite{Boyd-CVXBook} The optimal state $\mb z^\star$ of the constrained problem~\eqref{eq_dra2}-\eqref{eq_b2} satisfies $\nabla_z F(\mb z^\star) \in \mathrm{span}(\mb{s}_n)$. Equivalently, there exists $\psi^\star \in \mathbb{R}$ such that $b_i \nabla_i(z_i^\star) = \psi^\star$ for all $i$.
\end{lem}

\subsection{Quantized Equilibrium Characterization}

\begin{lem}[Quantized equilibrium] \label{lem_quantized_equilibrium}
	Let Assumption~\ref{ass_conn} hold. If $(\mb{z}^\infty, \mb{w}^\infty)$ is an equilibrium of the dynamics~\eqref{eq_sol_delay1}-\eqref{eq_sol_delay2}, then:
	\begin{align}
		\mb{w}^\infty &= \mb{0}, \label{eq_w_eq_zero} \\
		q\!\left(B\nabla_z F(\mb{z}^\infty)\right) &= \alpha\,\mb{1}_n \label{eq_quantized_eq_char}
	\end{align}
	for some $\alpha \in \mathbb{R}$.
\end{lem}
\begin{proof}
	At equilibrium, \eqref{eq_sol_delay2} gives $\mb{w}^\infty = \mb{z}^\infty - \mb{z}^\infty = \mb{0}$. Substituting into~\eqref{eq_sol_delay1}, using that all delayed copies coincide at equilibrium, yields $BL\,q(B\nabla_z F(\mb{z}^\infty)) = 0$. Because $B$ is invertible, $Lq(B\nabla_z F(\mb{z}^\infty))=0$. Connectivity gives $\ker L = \mathrm{span}(\mb{1}_n)$, and therefore $q(B\nabla_z F(\mb{z}^\infty))=\alpha\mb{1}_n$ for some $\alpha\in\mathbb R$.
\end{proof}

\begin{rem} \label{rem_quantized_eq}
	Lemma~\ref{lem_quantized_equilibrium} shows that \eqref{eq_quantized_eq_char} does \textit{not} imply $\nabla_z F(\mb{z}^\infty) \in \mathrm{span}(\mb{s}_n)$ and, therefore, does not imply exact optimality. Exact optimality is recovered only in the unquantized case $\rho = 0$.
\end{rem}

\subsection{Quantization-Error Bounds}

We retain the standard Laplacian bound only for the unquantized signal and handle quantization errors explicitly.

\begin{lem} \label{lem_xLy}
	Given a symmetric adjacency matrix $A$ with Laplacian $L$, define $\overline{\mb{z}} := \mb{z} - \frac{\mb{1}_n^\top \mb{z}}{n} \mb{1}_n$. Then
	\begin{align}
		\lambda_2 \|\overline{\mb{z}}\|_2^2 \leq \mb{z}^\top L \mb{z} = \overline{\mb{z}}^\top L \overline{\mb{z}} \leq \lambda_n \|\overline{\mb{z}}\|_2^2. \label{eq_laplace}
	\end{align}
\end{lem}
\begin{proof}
	Standard; see \cite{olfatisaberfaxmurray07}.
\end{proof}

\begin{lem}[Explicit feasible-sublevel gradient envelope] \label{lem_sublevel_envelope}
		Let $\mb z^\star$ minimize the smooth composite $F$ over $\mc Z_D$, let $B\nabla F(\mb z^\star)=\psi^\star\mathbf1$, and, for any chosen $R>0$, define
		\begin{align}
			\Omega_R&:=\{\mb z\in\mc Z_D:F(\mb z)-F(\mb z^\star)\leq R\},\label{eq_Omega_R}\\
			G_R&:=|\psi^\star|+2u\sqrt{R/v}. \label{eq_G_R}
		\end{align}
		Then $\Omega_R$ is compact and every $\mb z\in\Omega_R$ satisfies $\|B\nabla F(\mb z)\|_\infty\leq G_R$.
\end{lem}
\begin{proof}
		Strong convexity on the feasible affine space gives $F(\mb z)-F(\mb z^\star)\geq v\|\mb z-\mb z^\star\|^2$, so $\Omega_R$ is bounded and closed. Global $2u$-Lipschitz continuity of each scalar gradient and the KKT identity give
		$|b_i\nabla_i-\psi^\star|\leq2u|z_i-z_i^\star|\leq2u\sqrt{R/v}$.
\end{proof}

\begin{lem}[Componentwise quantization error] \label{lem_qerror}
	Let $\boldsymbol{\phi} \in \mathbb{R}^n$, $\mb{g} := q(\boldsymbol{\phi})$, $\mb{e} := \mb{g} - \boldsymbol{\phi}$. Under~\eqref{eq_quan_relerr}, $|e_i| \leq \delta_\rho |\phi_i|$ for all $i$, where $\delta_\rho := e^{\rho/2}-1$. If $\|\boldsymbol{\phi}\|_\infty \leq {G_R}$, then
	\begin{align}
		\|\mb{e}\|^2 &\leq n\delta_\rho^2 {G_R^2}, \label{eq_qerror_norm} \\
		\mb{e}^\top L\mb{e} &\leq \lambda_n n\delta_\rho^2 {G_R^2}. \label{eq_qerror_L}
	\end{align}
\end{lem}
\begin{proof}
	From~\eqref{eq_quan_relerr}, $|e_i| \leq \delta_\rho |\phi_i|$ for every $i$, including the case $\phi_i=0$ since $q(0)=0$. Therefore $|e_i| \leq \delta_\rho {G_R}$, which yields~\eqref{eq_qerror_norm} after squaring and summing. The bound~\eqref{eq_qerror_L} follows from $\mb{e}^\top L\mb{e} \leq \lambda_n\|\mb{e}\|^2$.
\end{proof}

\begin{lem}[Lower bound for the mixed Laplacian term] \label{lem_phiLg}
	Let $\boldsymbol{\phi} \in \mathbb{R}^n$, $\mb{g} = q(\boldsymbol{\phi})$, $\Phi := \boldsymbol{\phi}^\top L\boldsymbol{\phi}$. {For $z\in\Omega_R$} and every $\xi > 0$, we have
	\begin{align} \label{eq_phiLg_lower}
		\boldsymbol{\phi}^\top L\mb{g} \geq \left(1 - \frac{\xi}{2}\right)\Phi - \frac{\lambda_n n}{2\xi}\delta_\rho^2 {G_R^2}.
	\end{align}
\end{lem}
\begin{proof}
	Write $\mb{g} = \boldsymbol{\phi} + \mb{e}$. Then, it readily follows that $\boldsymbol{\phi}^\top L\mb{g} = \Phi + \boldsymbol{\phi}^\top L\mb{e}$. By Cauchy\nobreakdash-Schwarz and Young, we have
	\[
	|\boldsymbol{\phi}^\top L\mb{e}| \leq \sqrt{\Phi}\sqrt{\mb{e}^\top L\mb{e}} \leq \frac{\xi}{2}\Phi + \frac{1}{2\xi}\mb{e}^\top L\mb{e}.
	\]
	Applying Lemma~\ref{lem_qerror} yields~\eqref{eq_phiLg_lower}.
\end{proof}

\begin{lem}[Quantized Laplacian-energy bound] \label{lem_gLg}
	Let $\boldsymbol{\phi} \in \mathbb{R}^n$, $\mb{g} = q(\boldsymbol{\phi})$. {For $z\in\Omega_R$}, the following inequalities hold:
	\begin{align}
		\mb{g}^\top L\mb{g} &\leq 2\,\boldsymbol{\phi}^\top L\boldsymbol{\phi} + 4d_{\max} n\delta_\rho^2 {G_R^2}, \label{eq_gLg} \\
		\mb{g}^\top L^2\mb{g} &\leq 2\lambda_n\,\boldsymbol{\phi}^\top L\boldsymbol{\phi} + 4\lambda_n d_{\max} n\delta_\rho^2 {G_R^2}. \label{eq_gL2g}
	\end{align}
\end{lem}
\begin{proof}
	Let $\mb{e} := \mb{g} - \boldsymbol{\phi}$. Using $(a+b)^2 \leq 2a^2 + 2b^2$ and $(e_i - e_j)^2 \leq 2e_i^2 + 2e_j^2$, we obtain
	\[
	\mb{g}^\top L\mb{g} \leq 2\boldsymbol{\phi}^\top L\boldsymbol{\phi} + 4\sum_{i=1}^n d_i e_i^2.
	\]
	Lemma~\ref{lem_qerror} and $d_i \leq d_{\max}$ give~\eqref{eq_gLg}. Since $\mb{g}^\top L^2\mb{g} \leq \lambda_n\mb{g}^\top L\mb{g}$, \eqref{eq_gL2g} follows.
\end{proof}

\subsection{Descent Lemma}

\begin{lem}[Descent bound] \label{lem_strict}
	\cite{Boyd-CVXBook} Under Assumption~\ref{ass_cvx}, let $\delta(k) := \mb{z}(k+1) - \mb{z}(k)$. Then, we have
	\begin{align}
		F(\mb{z}(k+1)) \leq F(\mb{z}(k)) + \nabla_z F(\mb{z}(k))^\top \delta(k) + u\,\delta(k)^\top \delta(k). \label{eq_taylor_2}
	\end{align}
\end{lem}

\subsection{Anytime Feasibility}

\begin{lem}[Anytime feasibility] \label{lem_feas}
	Under Algorithm~\ref{alg_1} with initialization~\eqref{eq_init} and Assumption~\ref{ass_delay} (symmetric delays), the coupling constraint $\sum_{i=1}^n b_i z_i(k) = D$ holds for all $k \geq 0$.
\end{lem}
\begin{proof}
	Summing~\eqref{eq_sol_delay1} over all nodes weighted by $b_i$, the gradient-tracking term vanishes by antisymmetry ($A_{ij} = A_{ji}$ and $\mathcal{I}_{k,ij}(r) = \mathcal{I}_{k,ji}(r)$ from symmetric delays). The identity $\sum_{i=1}^n b_i w_i(k) = 0$ follows by induction from $w_i(0) = 0$. Therefore $\sum_{i=1}^n b_i z_i(k+1) = \sum_{i=1}^n b_i z_i(k) = D$.
\end{proof}

\begin{rem}
	The proof critically uses $\tau_{ij}(k) = \tau_{ji}(k)$. This property contrasts with ADMM-based approaches \cite{zhu2025deed,10740459,liu2023distributed0,han2024distributed}, where feasibility is achieved only asymptotically.
\end{rem}

\subsection{Feasible-Subspace Error Bound and Quantized Optimality Gap}

\begin{lem}[Feasible-subspace error bound] \label{lem_error_bound}
	Under Assumption~\ref{ass_cvx}, let $\mb z^\star$ be the optimizer of $F$ over $\mc{Z}_D := \{\mb{z} : \mb{s}_n^\top\mb{z} = D\}$, and define $\overline F(\mb z):=F(\mb z)-F(\mb z^\star)$ and $\overline{\boldsymbol{\phi}}(\mb{z}) := (I - \frac{1}{n}\mb{1}_n\mb{1}_n^\top)\boldsymbol{\phi}(\mb{z})$ with $\boldsymbol{\phi} = B\nabla_z F$. Then, every $\mb{z} \in \mc{Z}_D$ satisfies
	\begin{align} \label{eq_error_bound}
		\overline{F}(\mb{z}) \leq \frac{u}{4v^2}\|\overline{\boldsymbol{\phi}}(\mb{z})\|^2 \leq \frac{u}{4v^2\lambda_2}\Phi(\mb{z})
	\end{align}
	for the connected Laplacian $L$.
\end{lem}
\begin{proof}
	Let $\mb{x} := \mb{z} - \mb z^\star$ with $\mb{s}_n^\top\mb{x} = 0$, $\mb{y} := B\mb{x}$ with $\mb{1}_n^\top\mb{y} = 0$. By the mean-value theorem with diagonal Hessian $2vI \preceq H \preceq 2uI$ and $BHB = H$: $\overline{\boldsymbol{\phi}} = (I - \frac{1}{n}\mb{1}_n\mb{1}_n^\top)H\mb{y}$. Since $\mb{y} \perp \mb{1}_n$: $\mb{y}^\top\overline{\boldsymbol{\phi}} \geq 2v\|\mb{y}\|^2$. Cauchy\nobreakdash-Schwarz gives $\|\mb{x}\| \leq \frac{1}{2v}\|\overline{\boldsymbol{\phi}}\|$, then $\overline{F} \leq u\|\mb{x}\|^2 \leq \frac{u}{4v^2}\|\overline{\boldsymbol{\phi}}\|^2$.
\end{proof}

\begin{prop}[Approximate optimality of quantized equilibria] \label{prop_gap_quantized_eq}
	Let Assumptions~\ref{ass_cvx} and \ref{ass_conn} hold. Let $(\mb{z}^\infty, \mb{0})$ be an equilibrium of Algorithm~\ref{alg_1} {with $z^\infty\in\Omega_R$}. Then, $\mb{z}^\infty$ is feasible and there exists $\alpha \in \mathbb{R}$ such that $q(B\nabla_z F(\mb{z}^\infty)) = \alpha\mb{1}_n$. Moreover, defining
	\begin{align}
		\Delta_\rho := (e^\rho - 1){G_R}, \label{eq_delta_rho}
	\end{align}
	we have
	\begin{align}
		\max_{i,j}|\phi_i(\mb{z}^\infty) - \phi_j(\mb{z}^\infty)| &\leq \Delta_\rho, \label{eq_delta_rho_range} \\
		\|\mb{z}^\infty - \mb z^\star\| &\leq \frac{\sqrt{n}}{4v}\,\Delta_\rho, \label{eq_dist_gap} \\
		F(\mb{z}^\infty) - F(\mb z^\star) &\leq \frac{un}{16v^2}\,\Delta_\rho^2. \label{eq_cost_gap}
	\end{align}
	Hence, every fixed-quantized equilibrium has an $O(\rho^2)$ optimality gap.
\end{prop}
\begin{proof}
	Feasibility follows from Lemma~\ref{lem_feas}, and the quantized-consensus characterization follows from Lemma~\ref{lem_quantized_equilibrium}. If $\alpha = 0$, then $q(\phi_i^\infty)=0$ for every $i$, which for the logarithmic quantizer implies $\phi_i^\infty=0$ and the claims are immediate. Assume therefore that $\alpha \neq 0$. By~\eqref{eq_quan_sect},
	\[
	|\alpha|e^{-\rho/2} \leq |\phi_i^\infty| \leq |\alpha|e^{\rho/2}, \qquad \forall i.
	\]
	Since $|\phi_i^\infty| \leq {G_R}$, we also have $|\alpha| \leq e^{\rho/2}{G_R}$. Therefore
	\[
	\max_{i,j}|\phi_i^\infty - \phi_j^\infty|
	\leq |\alpha|\bigl(e^{\rho/2}-e^{-\rho/2}\bigr)
	\leq {G_R}(e^\rho - 1)
	= \Delta_\rho,
	\]
	which is~\eqref{eq_delta_rho_range}. Thus, we obtain $\|\overline{\boldsymbol{\phi}}(\mb{z}^\infty)\|^2 \leq \frac{n}{4}\Delta_\rho^2$. Applying Lemma~\ref{lem_error_bound} yields~\eqref{eq_dist_gap} and~\eqref{eq_cost_gap}.
\end{proof}

\subsection{Practical Convergence Under Fixed Quantization} \label{sec_convergence}

We now state the main result for the fixed log-quantized dynamics. Recall from Remark~\ref{rem_penalty} that $F$, $f_i$, and $u$ refer to the penalized cost. {Throughout this subsection $\mc G=(\mc V,\mc E)$ is static, connected, and undirected, with fixed Laplacian $L$, and $\overline F(k):=F(\mb z(k))-F(\mb z^\star)$, where $\mb z^\star$ is the unique optimizer of the smooth penalized problem over $\mc Z_D$.} Consider the following definitions:
\begin{align*}
	\boldsymbol{\phi}(k) &:= B\nabla_z F(\mb{z}(k)), \qquad
	\mb{g}(k) := q(\boldsymbol{\phi}(k)),\\
	\Phi(k) &:= \boldsymbol{\phi}(k)^\top L\boldsymbol{\phi}(k),\qquad
	\delta(k) := \mb{z}(k+1)-\mb{z}(k),\\
	H(k)&:=\sum_{s=k-\bar\tau}^{k-1}\|\delta(s)\|^2,\\
	W(k)&:=\sum_{j=1}^{\bar\tau}\ \sum_{s=k-j}^{k-1}\|\delta(s)\|^2,
\end{align*}
with $H(k)=W(k)=0$ when $\bar\tau=0$. {For $\bar\tau\geq1$, $H(k)\leq W(k)\leq\bar\tau H(k)$ and
	$W(k+1)-W(k)=\bar\tau\|\delta(k)\|^2-H(k)$.}

\begin{lem}[Correct delay-mismatch bound] \label{lem_delay_mismatch}
		Choose an orientation $h_e=e_i-e_j$ for each edge $e=\{i,j\}$, let $a_e=A_{ij}$ and $r_e(k)=\tau_{ij}(k)$, and define
		\begin{align}
			p(k):=B\sum_{e\in\mc E}a_eh_eh_e^\top\{g(k)-g(k-r_e(k))\}. \label{eq_p_def}
		\end{align}
		If $z(k-r),\,r=0,\ldots,\bar\tau$, lie in $\Omega_R$, then
		\begin{align}
			\|p(k)\|^2&\leq C_HH(k)+C_{p,\rho}, \label{eq_p_correct}\\
			C_H&:=16u^2\lambda_nd_{\max}\bar\tau,\nonumber\\
			C_{p,\rho}&:=16\chi_{\bar\tau}\lambda_nnd_{\max}\delta_\rho^2G_R^2, \label{eq_p_constants}
		\end{align}
		where $\chi_{\bar\tau}=0$ if $\bar\tau=0$ and $1$ otherwise.
\end{lem}
\begin{proof}
		Write $g=\phi+e_q$ and split $p=p_\phi+p_e$. The incidence-matrix inequality
		$\|\sum_ea_eh_e\alpha_e\|^2\leq\lambda_n\sum_ea_e\alpha_e^2$ applies to both parts. Global smoothness gives
		$|\phi_i(k)-\phi_i(k-r)|^2\leq4u^2\bar\tau\sum_{s=k-\bar\tau}^{k-1}|\delta_i(s)|^2$; summing endpoints with weighted degree at most $d_{\max}$ yields $\|p_\phi\|^2\leq8u^2\lambda_nd_{\max}\bar\tau H(k)$. Lemmas~\ref{lem_sublevel_envelope}-\ref{lem_qerror} give $|e_{q,i}(t)|\leq\delta_\rho G_R$, and hence $\|p_e\|^2\leq8\chi_{\bar\tau}\lambda_nnd_{\max}\delta_\rho^2G_R^2$. Finally, $\|p_\phi+p_e\|^2\leq2\|p_\phi\|^2+2\|p_e\|^2$ proves~\eqref{eq_p_correct}.
\end{proof}

\begin{rem}
	{Lemma~\ref{lem_delay_mismatch} does not use an estimate of the form $\|q(\phi(k))-q(\phi(k-r))\|\leq C\|\phi(k)-\phi(k-r)\|$; the discontinuous quantization error is bounded additively.}
\end{rem}

\begin{theorem}[Ergodic practical-stability under delays and momentum] \label{thm_conv}
	Let Assumptions~\ref{ass_cvx}, \ref{ass_conn}, and~\ref{ass_delay} hold on a static connected undirected graph. Fix $R>0$, compute $G_R$ from~\eqref{eq_G_R}, and suppose the feasible initialization and constant prehistory satisfy $\mc V(0)\leq R$. Consider the Lyapunov\nobreakdash-Krasovskii functional
	\begin{align} \label{eq_lyap_full}
		\mc{V}(k) := \overline{F}(k) + \frac{{\omega_w}}{2}\|\mb{w}(k)\|^2 + {\omega_H}\,W(k),
	\end{align}
	with {Lyapunov\nobreakdash-Krasovskii weights $\omega_w,\omega_H>0$}, Young parameters $\epsilon, \gamma > 0$, and auxiliary parameter $\xi \in (0,2)$. Define {$C_H,C_{p,\rho}$ by~\eqref{eq_p_constants}} and the quantization-induced constants
	\begin{align}
		C_{\times,\rho} &:= \frac{\lambda_n n}{2\xi}\delta_\rho^2 {G_R^2}, \quad C_{g,\rho} := 4\lambda_n d_{\max} n\delta_\rho^2 {G_R^2}, \label{eq_Crho_consts}
	\end{align}
	and the descent/residual coefficients
	\begin{align}
		\wt{\beta}_1 &:= \eta_\tau(1 - \xi/2) - \frac{\mu}{2\epsilon\lambda_2}
		- \frac{\eta_\tau}{2\gamma\lambda_2}\nonumber\\
		&\quad - (6u + 3{\omega_w} + 6{\omega_H}\bar\tau)\eta_\tau^2\lambda_n, \label{eq_beta1} \\
		\wt{\beta}_2 &:= \frac{{\omega_w}}{2}(1-3\mu^2) - 3u\mu^2 - \frac{\mu\epsilon}{2} - 3{\omega_H}\bar\tau\mu^2, \label{eq_beta2} \\
		\wt{\beta}_3 &:= {\omega_H} - \left(\frac{\eta_\tau\gamma}{2} + (3u+\tfrac{3{\omega_w}}{2}+3{\omega_H}\bar\tau)\eta_\tau^2\right){C_H}, \label{eq_beta3} \\
		\varepsilon_\rho &:= \eta_\tau C_{\times,\rho}\nonumber\\
		&\quad+\left(\frac{\eta_\tau\gamma}{2}+(3u+\tfrac{3{\omega_w}}{2}+3{\omega_H}\bar\tau)\eta_\tau^2\right){C_{p,\rho}}\nonumber\\
		&\quad+(3u + \tfrac{3{\omega_w}}{2} + 3{\omega_H}\bar\tau)\eta_\tau^2 C_{g,\rho}. \label{eq_epsrho}
	\end{align}
	A necessary condition for this sufficient certificate is $\mu<1/\sqrt3$, since $\wt\beta_2\leq \omega_w(1-3\mu^2)/2$. Define
		\begin{align}
			\kappa_H:=\begin{cases}
				\wt\beta_3/(\omega_H\bar\tau),&\bar\tau\geq1,\\
				+\infty,&\bar\tau=0,
			\end{cases} \label{eq_kappa_H}
		\end{align}
		and let
		\begin{align}
			A_F:=\frac{u}{4v^2\lambda_2},\qquad
			\kappa_V:=\min\left\{1,\frac{\wt\beta_1}{A_F},\frac{2\wt\beta_2}{\omega_w},\kappa_H\right\}. \label{eq_kappa_V}
		\end{align}
		If $\wt{\beta}_1,\wt{\beta}_2,\wt{\beta}_3>0$ and the self-consistency condition
		\begin{align}
			\varepsilon_\rho\leq\kappa_VR \label{eq_self_consistency}
		\end{align}
		holds, then $z(k)\in\Omega_R$ for every $k$ and
	\begin{align} \label{eq_practical_descent}
		\mc{V}(k+1) - \mc{V}(k) \leq -\wt{\beta}_1\Phi(k) - \wt{\beta}_2\|\mb{w}(k)\|^2 - \wt{\beta}_3 H(k) + \varepsilon_\rho.
	\end{align}
	Moreover,
		\begin{align}
			\mc V(k)&\leq(1-\kappa_V)^k\mc V(0)
			+\frac{1-(1-\kappa_V)^k}{\kappa_V}\varepsilon_\rho, \label{eq_V_ultimate}
		\end{align}
	Consequently,
	\begin{align} \label{eq_average_bound}
		&\limsup_{N\to\infty}\frac{1}{N}\sum_{k=0}^{N-1}
		\Big(\Phi(k) + \|\mb{w}(k)\|^2 + H(k)\Big)\nonumber\\
		&\hspace{2em}\leq \frac{\varepsilon_\rho}{\min\{\wt{\beta}_1, \wt{\beta}_2, \wt{\beta}_3\}}.
	\end{align}
	Hence, the fixed-quantized dynamics satisfies an explicit ergodic residual bound of order $O(\delta_\rho^2)=O(\rho^2)$.
\end{theorem}

\begin{proof}
	All iterates lie in $\mc{Z}_D$ by Lemma~\ref{lem_feas}. Assume inductively that the states used at time $k$ lie in $\Omega_R$; this is true at $k=0$ by the constant-prehistory convention. The delayed displacement decomposes as
		$\delta(k)=-\eta_\tau BLg(k)+\eta_\tau p(k)+\mu w(k)$, and Lemma~\ref{lem_delay_mismatch} gives $\|p(k)\|^2\leq C_HH(k)+C_{p,\rho}$.
	
	\textit{Step 1: cost increment.} By Lemma~\ref{lem_strict}: $\overline{F}(k+1) - \overline{F}(k) \leq \nabla_z F^\top\delta + u\|\delta\|^2$. The Laplacian term uses Lemma~\ref{lem_phiLg}: $-\eta_\tau\boldsymbol{\phi}^\top L\mb{g} \leq -\eta_\tau(1-\xi/2)\Phi + \eta_\tau C_{\times,\rho}$. The cross-terms use $\mb{s}_n^\top\mb{w} = 0$ and $\mb{s}_n^\top\mb{p} = 0$ (Lemma~\ref{lem_feas}):
	\begin{align}
		\mu|\nabla_z F^\top\mb{w}| &\leq \frac{\mu}{2\epsilon\lambda_2}\Phi + \frac{\mu\epsilon}{2}\|\mb{w}\|^2, \\
		\eta_\tau|\nabla_z F^\top\mb{p}| &\leq \frac{\eta_\tau}{2\gamma\lambda_2}\Phi + \frac{\eta_\tau\gamma}{2}{(C_H H+C_{p,\rho})}.
	\end{align}
	Using Lemma~\ref{lem_gLg}: $\mb{g}^\top L^2\mb{g} \leq 2\lambda_n\Phi + C_{g,\rho}$. Therefore, we obtain
	\begin{align}
		\|\delta\|^2 &\leq 3\eta_\tau^2(2\lambda_n\Phi + C_{g,\rho})\nonumber\\
		&\quad + 3\eta_\tau^2 {(C_HH+C_{p,\rho})} + 3\mu^2\|\mb{w}\|^2. \label{eq_delta_sq_quant}
	\end{align}
	
	\textit{Step 2: momentum increment.} Since $\mb{w}(k+1) = \delta(k)$, using~\eqref{eq_delta_sq_quant}:
	\begin{align} \nonumber
		\frac{{\omega_w}}{2}(\|\mb{w}(k+1)\|^2 - \|\mb{w}(k)\|^2)
		&\leq 3{\omega_w}\eta_\tau^2\lambda_n\Phi + \frac{3{\omega_w}}{2}\eta_\tau^2 C_{g,\rho} \\
		&\quad+ \frac{3{\omega_w}}{2}\eta_\tau^2 {(C_HH+C_{p,\rho})}\\
		&\quad- \frac{{\omega_w}}{2}(1-3\mu^2)\|\mb{w}\|^2.
	\end{align}
	
	\textit{Step 3: Krasovskii-history increment.} {The double sum satisfies $W(k+1)-W(k)=\bar\tau\|\delta(k)\|^2-H(k)$.} Hence,
	\begin{align} \nonumber
		{\omega_H}(W(k+1) - W(k)) &\leq 6{\omega_H}\bar\tau\eta_\tau^2\lambda_n\Phi
		+3{\omega_H}\bar\tau\eta_\tau^2C_{g,\rho}\nonumber\\
		&\quad+3{\omega_H}\bar\tau\eta_\tau^2{(C_HH+C_{p,\rho})}\nonumber\\
		&\quad+3{\omega_H}\bar\tau\mu^2\|\mb{w}\|^2-{\omega_H}H(k). \label{eq_W_increment}
	\end{align}
	
	\textit{Step 4: combined descent and invariant set.} Adding all three gives~\eqref{eq_practical_descent}. Lemma~\ref{lem_error_bound} gives $\overline F(k)\leq A_F\Phi(k)$, while $W(k)\leq\bar\tau H(k)$ for $\bar\tau\geq1$ (and both vanish when $\bar\tau=0$). Hence~\eqref{eq_kappa_H}-\eqref{eq_kappa_V} imply $\mc V(k+1)\leq(1-\kappa_V)\mc V(k)+\varepsilon_\rho$. If $\mc V(k)\leq R$, condition~\eqref{eq_self_consistency} gives $\mc V(k+1)\leq R$, so induction proves $z(k)\in\Omega_R$ and validates every bound used above without an orbit assumption. Iterating the scalar recursion gives~\eqref{eq_V_ultimate}. Summing~\eqref{eq_practical_descent} from $k = 0$ to $N-1$, dividing by $N$, and using $\mc{V}(k) \geq 0$ gives~\eqref{eq_average_bound}.
\end{proof}

\begin{rem}[Interpretation of Theorem~\ref{thm_conv}]
	Theorem~\ref{thm_conv} is an ergodic practical-stability result for the fixed-$\rho$ dynamics. It does not claim exact convergence of the state trajectory to the penalized optimizer when $\rho > 0$. Instead, it guarantees that the average gradient-disagreement, momentum energy, and delay-history energy remain bounded by an explicit constant of order $O(\rho^2)$. Exact convergence is recovered only in the unquantized limit $\rho = 0$; see Corollary~\ref{cor_exact_rho0}. {Existence of a fixed-quantized equilibrium is not asserted: discontinuous fixed quantization may instead yield a bounded cycle. Proposition~\ref{prop_gap_quantized_eq} and Corollary~\ref{cor_quantized_limit} are conditional equilibrium statements, whereas~\eqref{eq_average_bound} is unconditional under the theorem hypotheses.}
\end{rem}

\begin{cor}[Exact convergence in the unquantized case] \label{cor_exact_rho0}
	Under the assumptions of Theorem~\ref{thm_conv}, if $\rho = 0$, then $\delta_\rho = \varepsilon_\rho = 0$ and~\eqref{eq_practical_descent} is a strict Lyapunov descent. Consequently, $\sum_{k=0}^\infty \Phi(k) < \infty$, $\sum_{k=0}^\infty \|\mb{w}(k)\|^2 < \infty$, and $\sum_{k=0}^\infty H(k) < \infty$, hence
	\[
	\Phi(k)\to 0,\qquad \mb{w}(k)\to \mb{0},\qquad H(k)\to 0.
	\]
	Every accumulation point of $\mb{z}(k)$ is feasible by Lemma~\ref{lem_feas} and satisfies $\overline{\boldsymbol{\phi}}(\mb{z})=\mb{0}$ because $\Phi(k)\to 0$ and {$L$ is connected}. By Lemma~\ref{lem_z*}, every accumulation point therefore equals the unique optimizer of the penalized problem on $\mc{Z}_D$, denoted again by $\mb z^\star$. Hence, it readily follows that $\mb{z}(k)\to \mb z^\star$.
\end{cor}

\begin{cor}[Optimality of convergent quantized trajectories] \label{cor_quantized_limit}
	Under the assumptions of Theorem~\ref{thm_conv}, if the trajectory converges to an equilibrium $(\mb{z}^\infty, \mb{0})$, then
	\[
	F(\mb{z}^\infty) - F(\mb z^\star) \leq \frac{un {G_R^2}}{16v^2}(e^\rho - 1)^2.
	\]
\end{cor}
\begin{proof}
	Apply Proposition~\ref{prop_gap_quantized_eq}.
\end{proof}

\begin{rem}[Verifying the descent conditions]
	The conditions $\wt{\beta}_1, \wt{\beta}_2, \wt{\beta}_3 > 0$ can be satisfied for sufficiently small $\eta_\tau$ and $\mu$: the leading term in $\wt{\beta}_1$ is $O(\eta_\tau)$, the negative corrections are $O(\eta_\tau^2)$ and $O(\mu)$; $\wt{\beta}_2 > 0$ requires $\mu$ small enough; $\wt{\beta}_3 > 0$ requires ${\omega_H}$ large relative to $\eta_\tau C_H$. {The necessary ceiling $\mu<1/\sqrt3$ is only algebraic and need not be attainable after the coupled $\wt\beta_1$ and $\wt\beta_3$ restrictions are imposed. Experiments outside the explicitly computed admissible region are labeled empirical.}
\end{rem}

\begin{rem}[On the admissible $(\eta_\tau, \mu)$ region] \label{rem_momentum}
	Conditions~\eqref{eq_beta1}-\eqref{eq_beta3} define a two-dimensional sufficient region. {Proposition~\ref{prop_spectral_hb} gives the exact optimum for the delay-free quadratic baseline. Identifying the optimal $\mu^*(\eta_\tau,\bar\tau,\rho)$ for the general robust setting remains open.}
\end{rem}

\begin{rem}[Delay-step-size trade-off] \label{rem_delay_tradeoff}
	The bound on $\eta_\tau$ from $\wt{\beta}_3 > 0$ scales as $\mc{O}(1/(\bar\tau+1))$. {For small $\eta_\tau$, its leading form is $\eta_\tau<2\omega_H/(\gamma C_H)=\omega_H/(8\gamma u^2\lambda_nd_{\max}\bar\tau)$ when $\bar\tau\geq1$.}
\end{rem}

\begin{rem}[Network topology] \label{rem_topology}
	The bounds depend on $\lambda_2/\lambda_n$ and $d_{\max}$. Graphs with higher algebraic connectivity admit larger step sizes and smaller residuals \cite{qureshi2023distributed,qureshi2020s,li2022implicit}. {Table~\ref{tab_tradeoffs} summarizes the leading effects of the principal parameters.}
\end{rem}

\section{Simulation Results} \label{sec_sim}

This section validates the theoretical results through numerical experiments. All simulations are performed on an Intel Core-i5 laptop (2.50\,GHz, 8\,GB RAM). Following \cite{wood2013power,yang2013consensus}, we use quadratic costs:
\begin{align} \label{eq_f_quad}
	f_i(z_i) = \gamma_i z_i^2 + \beta_i z_i + \alpha_i,
\end{align}
with parameters in Table~\ref{tab_par}. {For the unpenalized quadratics, $v_0 = \min_i \gamma_i = 0.03$ and $u_0 = \max_i \gamma_i = 0.04$. Whenever the smooth penalty is used, the theorem must instead use the global composite constant specified in Remark~\ref{rem_penalty}.}

\begin{table}[hbpt!]
	\centering
	\caption{Power cost parameters for different energy node types \cite{wood2013power,yang2013consensus}.}
	\begin{tabular}{|c|c|c|c|}
		\hline
		Type & $\alpha_i$ (\$/h) & $\beta_i$ (\$/MWh) & $\gamma_i$ (\$/MW$^2$h) \\
		\hline
		A & 561 & 2.0 & 0.04 \\
		\hline
		B & 310 & 3.0 & 0.03 \\
		\hline
		C & 78 & 4.0 & 0.035 \\
		\hline
		D & 561 & 4.0 & 0.03 \\
		\hline
		E & 78 & 2.5 & 0.04 \\
		\hline
	\end{tabular}
	\label{tab_par}
\end{table}

\subsection{Effect of the Momentum Coefficient} \label{sec_sim_mu}

We consider the delay-free case ($\bar\tau = 0$) with $n = 40$ energy nodes (8 per type) over a sparse Erd\H{o}s\nobreakdash-R\'{e}nyi (ER) graph with $15\%$ link probability and random symmetric weights. The topology is redrawn every 100 iterations. The smooth penalty~\eqref{eq_log_penalty} is used with $m_i = 10$, $M_i = 110$, $\sigma = 2$, $\kappa = 2$. Demand is $D = 2000$\,MW with $d_i = 50$ for all nodes, $\eta = 0.05$, and $\rho = 0.0078$.

	Because this experiment redraws the topology, it is not a validation of the static-graph Theorem~\ref{thm_conv}. It is retained as an empirical topology-variation study. The previous step \nobreakdash-size ``verification,'' which used $u_0=0.04$ while omitting the smooth-penalty curvature and the other theorem coefficients, has been removed.

Fig.~\ref{fig_momentum} shows results for $\mu = 0, 0.25, 0.5, 0.75$. Larger $\mu$ yields faster decay of the cost residual. The weighted average $\frac{1}{n} \sum_{i=1}^n b_i z_i(k)$ remains constant, verifying anytime equality feasibility (Lemma~\ref{lem_feas}). The momentum variables $w_i(k)$ decay toward zero. {These are empirical curves: in particular, $\mu=0.75>1/\sqrt3$ cannot satisfy the sufficient condition~\eqref{eq_beta2}. Proposition~\ref{prop_spectral_hb}, rather than the robust theorem, explains the acceleration mechanism in the static unquantized quadratic baseline.}

\begin{figure}[!]
	\centering
	\includegraphics[width=2.45in]{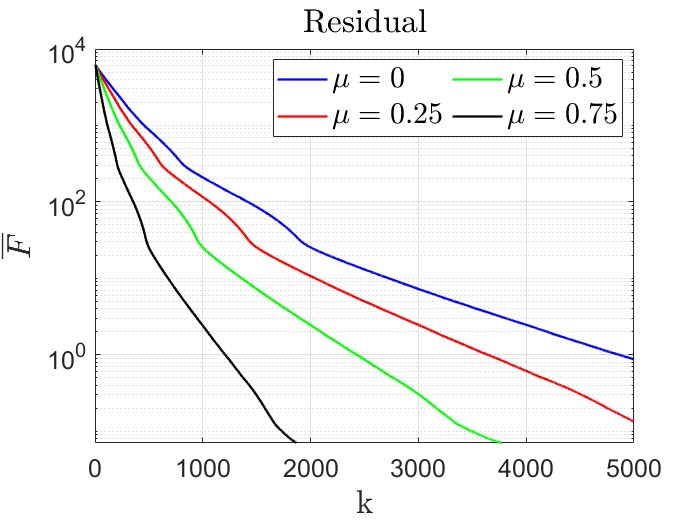}
	\includegraphics[width=2.45in]{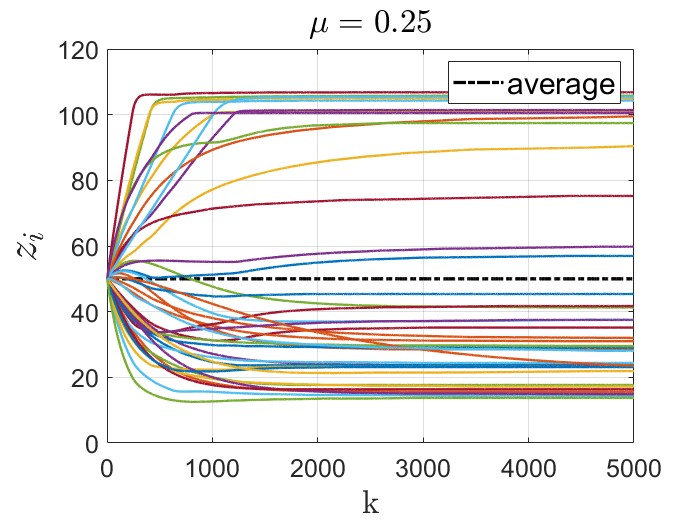}
	\includegraphics[width=2.45in]{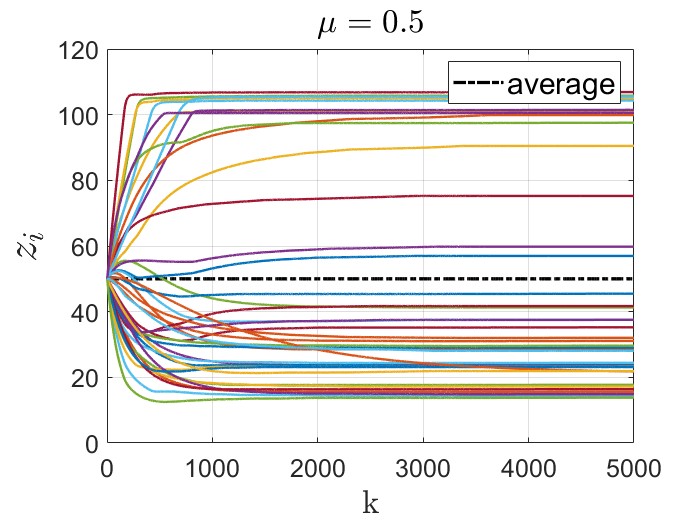}
	\includegraphics[width=2.45in]{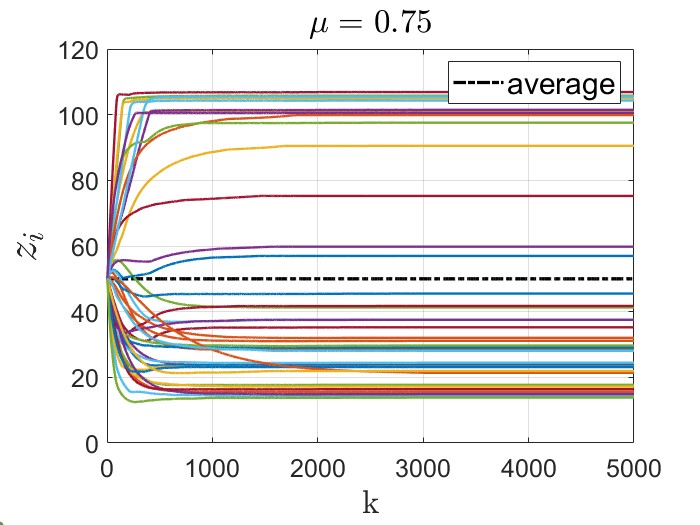}
	\includegraphics[width=2.45in]{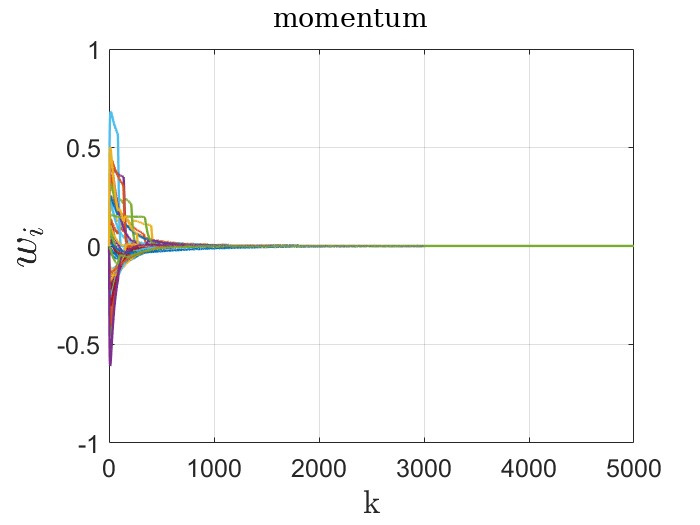}
	\caption{Effect of momentum $\mu$ on convergence ($n=40$, delay-free). Top-left: cost residual for $\mu = 0, 0.25, 0.5, 0.75$. Others: power allocations $z_i(k)$ showing constant weighted average (anytime feasibility), and momentum variables $w_i(k)$ converging to zero.}
	\label{fig_momentum}
\end{figure}

\subsection{Spectral-Rate Validation on a Fixed Sparse Graph} \label{sec_spectral_validation}

	To test the quantitative prediction in Proposition~\ref{prop_spectral_hb}, we add a static, delay-free, unquantized quadratic benchmark. It uses $n=40$, $B=I$, the five quadratic cost types in Table~\ref{tab_par} repeated eight times, and one connected weighted ER graph with link probability $0.15$, random seed $260826$, and edge weights sampled uniformly from $[0.5,1.5]$. The resulting graph has 125 edges. For its actual $Q$ and $L$,
	\[
	\theta_{\min}=0.0877407,\quad \theta_{\max}=1.119718,\quad \varkappa=12.7617.
	\]
	Equations~\eqref{eq_hb_opt}-\eqref{eq_hb_rate} therefore give $(\eta_{\rm HB},\mu_{\rm HB})=(2.18062,0.316505)$ and $r_{\rm HB}=0.562588$. The optimal zero-momentum values are $\eta_{\rm GD}=1.65637$ and $r_{\rm GD}=0.854669$. At $(\eta_{\rm HB},\mu_{\rm HB})$, the discriminant in~\eqref{eq_hb_mode} is zero at both $\theta_{\min}$ and $\theta_{\max}$ (to floating-point precision), so the extreme-mode responses can contain a factor linear in $k$. Accordingly, Fig.~\ref{fig_spectral_rate} does not use a bare $r_{\rm HB}^k$ upper bound: its dashed heavy\nobreakdash-ball curve is the exact zero-initial-momentum worst-mode recurrence bound over the actual positive spectrum, including the repeated-root polynomial factor, while $r_{\rm HB}$ labels only the asymptotic root modulus. The observed curves remain below the exact modal bounds and display the predicted acceleration. This experiment validates Proposition~\ref{prop_spectral_hb}; it is separate from the delayed quantized certificate in Section~\ref{sec_sim_delay_certified}.

\begin{figure}[t]
	\centering
	\includegraphics[width=3.75in]{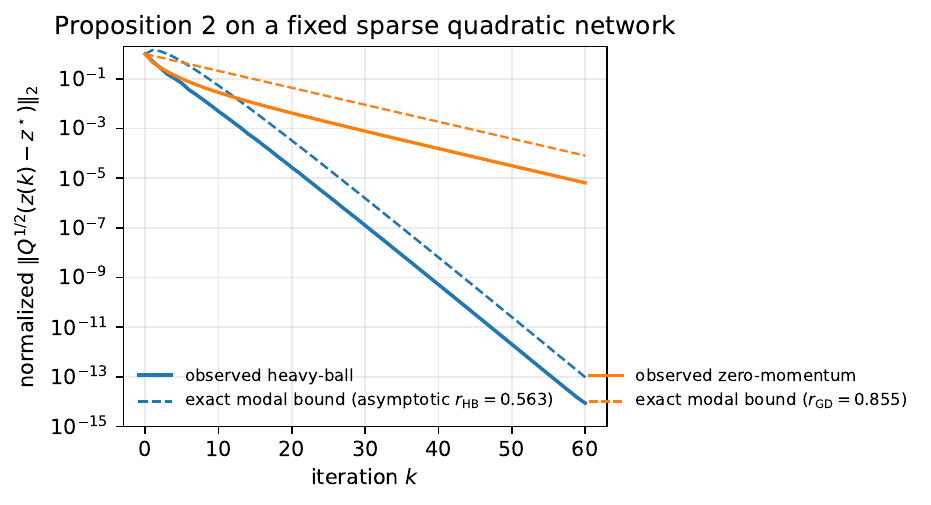}
	\caption{{Observed heavy-ball and zero-momentum errors on a fixed sparse quadratic network. Dashed curves are exact zero-initial-momentum worst-mode recurrence bounds computed from the actual positive spectrum of $Q^{1/2}LQ^{1/2}$; the displayed $r_{\rm HB}$ is an asymptotic root-modulus factor, and the heavy-ball bound retains the polynomial factor caused by repeated roots at the two extreme modes.}}
	\label{fig_spectral_rate}
\end{figure}

\subsection{Mixed Generator-Battery Signed Validation} \label{sec_mixed_signed}

	To exercise the signing matrix numerically, we add a fixed delay-free quantized instance with $n_1=12$ generator nodes, $n_2=8$ battery nodes, and $B=\operatorname{diag}(\mathbf1_{12},-\mathbf1_8)$. The connected weighted ER graph has 20 nodes, 36 edges, link probability $0.20$, and random seed $260827$. We use strongly convex quadratic costs, $z_i(0)=40$, $D=b^\top z(0)=160$, $(\eta,\mu,\rho)=(0.05,0.40,0.01)$, and 2000 iterations. Fig.~\ref{fig_mixed_signed} shows the cost gap and signed-gradient disagreement decreasing while the signed equality is conserved: the final cost gap is $4.7621\times10^{-3}$, the final signed-gradient disagreement is $4.0625\times10^{-2}$, and the maximum numerical error in $b^\top z(k)=D$ is $4.21\times10^{-12}$. At $k=2000$, the quantized signed gradients agree to machine precision and the momentum vanishes. Taking $R=F(z(0))-F(z^\star)=273.7169$ gives $G_R=12.4751$ and $\Delta_\rho=0.12538$, whereas the observed maximum pairwise signed-gradient spread is $0.02394$. Thus, the terminal numerical equilibrium satisfies the bound in Proposition~\ref{prop_gap_quantized_eq}; its nonzero disagreement plateau is consistent with Remark~\ref{rem_quantized_eq}. This is an equilibrium-bound check and empirical sign-pattern validation, not an additional Theorem~\ref{thm_conv} certificate, but unlike the generator-only studies it directly exercises both occurrences of $B$ in~\eqref{eq_L}.

\begin{figure*}[t]
	\centering
	\includegraphics[width=0.95\textwidth]{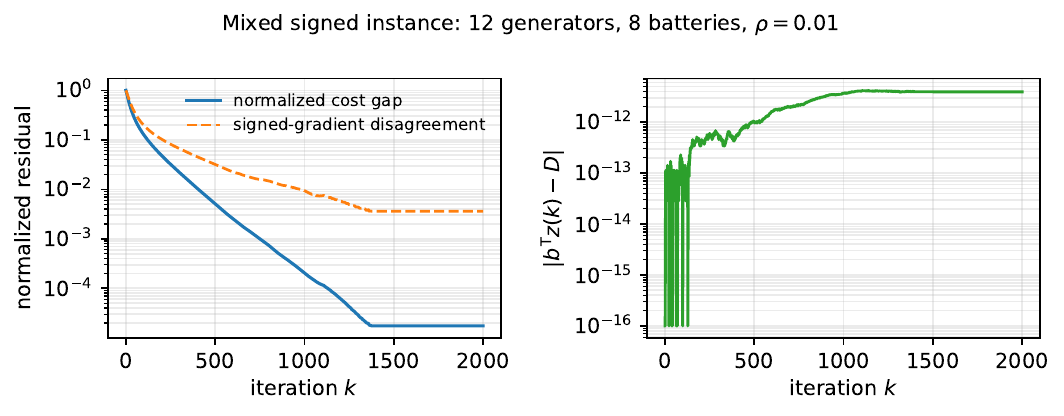}
	\caption{{Empirical mixed signed validation with 12 generators ($b_i=+1$) and 8 batteries ($b_i=-1$) on a fixed sparse graph. Left: normalized cost gap and signed-gradient disagreement. Right: numerical error in the conserved coupling equality $b^\top z(k)=D$. The gradual increase in the equality-error trace remains at floating-point scale and reflects accumulated roundoff in an algebraically cancelling update, not loss of the conservation invariant.}}
	\label{fig_mixed_signed}
\end{figure*}

\subsection{Effect of the Quantization Resolution} \label{sec_sim_quant}

All parameters, {including the topology redrawn every 100 iterations}, are as in Section~\ref{sec_sim_mu}, with $\mu = 0.5$ and $\rho = \frac{1}{16}, \frac{1}{64}, \frac{1}{128}, \frac{1}{256}$. {Accordingly, this is also an empirical switching-topology study outside Theorem~\ref{thm_conv}.} Fig.~\ref{fig_sim_quant} shows stable behavior for all levels, with finer quantization yielding smaller steady-state residuals. The $O(\delta_\rho^2)$ scaling of~\eqref{eq_epsrho} motivates this comparison, but the switching-topology curves are not used as theorem validation. Fig.~\ref{fig_quant_comp} compares log-scale and uniform quantization \cite{10605038} at $\rho = 0.0156$; uniform quantization produces a substantially larger residual due to its dead-zone.

\begin{figure}[t]
	\centering
	\includegraphics[width=3in]{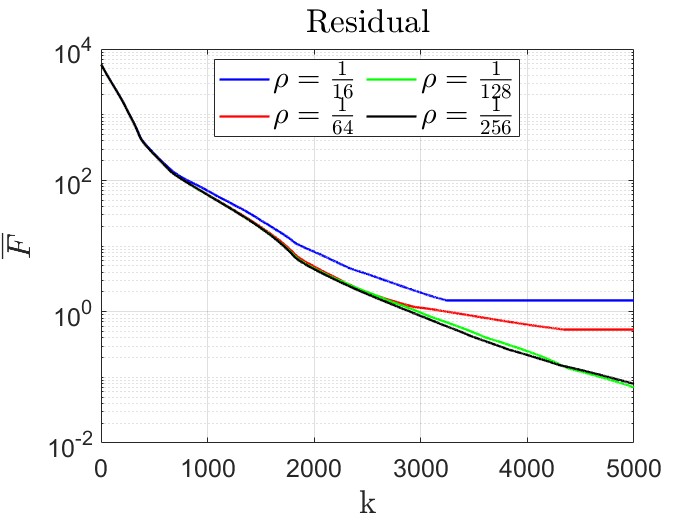}
	\caption{Cost residual under different log-scale quantization levels $\rho$.}
	\label{fig_sim_quant}
\end{figure}

\begin{figure}[t]
	\centering
	\includegraphics[width=3in]{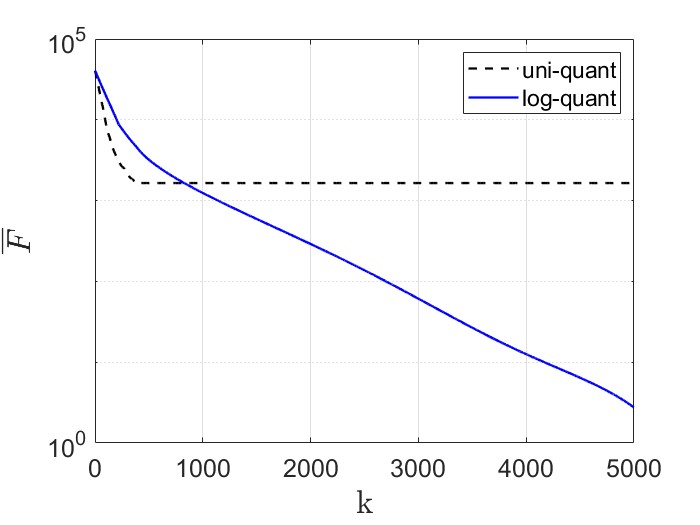}
	\caption{Comparison of uniform \cite{10605038} and log-scale quantization at $\rho = 0.0156$.}
	\label{fig_quant_comp}
\end{figure}

\subsection{Comparison with Existing Methods} \label{sec_sim_compare}

We compare Algorithm~\ref{alg_1} ($\mu = 0.9$) against: linear \cite{cherukuri2015distributed}, accelerated \cite{guo2019accelerated} ($\alpha = 0.2$, $\beta = 0.9$), single-bit \cite{taes}, fixed\nobreakdash-time \cite{liu2023distributed} ($\zeta = 0.25$), and saturated \cite{dsp} ($\delta = 1$). {Each comparison method uses the parameter values from its original reference and was not retuned for this instance. Under these stated settings, Fig.~\ref{fig_compare} shows that Algorithm~\ref{alg_1} exhibits the fastest residual decay. The figure therefore compares the stated settings rather than the best achievable performance of each method.} Table~\ref{tab_comp} reports comparable wall-clock times. {This is a delay-free empirical comparison; $\mu=0.9$ is outside Theorem~\ref{thm_conv}, and no claim is made that the comparison methods or the broader literature cannot handle delays. Proposition~\ref{prop_spectral_hb} does permit any $\mu<1$, including $0.9$, for a sufficiently small step size in the static delay-free unquantized quadratic baseline; the present switching/quantized comparison nevertheless remains empirical. Recent delay-robust alternatives are discussed in Section~\ref{sec_intro}. We do not present a direct numerical comparison with the continuous-time scattering and frequency-control schemes in \cite{shi2025delayed,li2023delay}, because their plant models, state variables, and delay guarantees differ materially; constructing a common benchmark would require an additional modeling study.}

\begin{table*}[hbpt!]
	\centering
	\caption{Wall-clock time (5000 iterations) for different DERMS algorithms.}
	\begin{tabular}{|c|c|c|c|c|}
		\hline
		Algorithm & Run Time & Algorithm & Run Time \\ \hline Linear \cite{cherukuri2015distributed} & 0.742\,s & Accelerated \cite{guo2019accelerated} & 0.797\,s \\ \hline Single-bit \cite{taes} & 0.791\,s & Fixed-time \cite{liu2023distributed} & 1.147\,s\\ \hline   Saturated \cite{dsp} & 0.926\,s & This work & 0.747\,s \\
		\hline
		\hline
	\end{tabular}
	\label{tab_comp}
\end{table*}

\begin{figure}[t]
	\centering
	\includegraphics[width=3in]{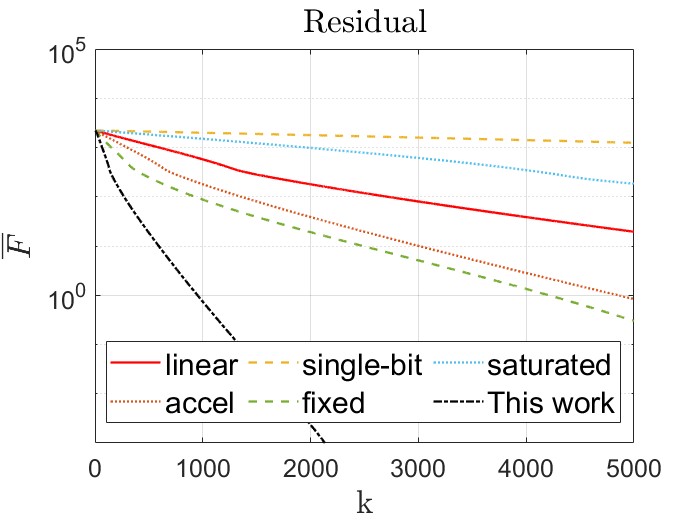}
	\caption{Cost residual comparison with existing DERMS methods.}
	\label{fig_compare}
\end{figure}

\subsection{Corrected Delayed Certificate} \label{sec_sim_delay_certified}

	The previously submitted parameters $(\mu,\eta_\tau,\rho)=(0.4,0.03,0.0156)$ were incorrectly described as certified and have been withdrawn. We instead use a static complete graph of $n=25$ generator nodes with $A_{ij}=1/n$, the five cost types in Table~\ref{tab_par} repeated five times, $m_i=10$, $M_i=110$, $(\sigma,\kappa)=(2,2)$, $D=1000$ MW, and $z_i(0)=40$. The paired edge delays are sampled in $\{0,1,2\}$. Equation~\eqref{eq_penalty_u_sharp} gives $u<0.55$, so $v=0.03$ and $u=0.55$ are valid global composite constants. {The penalized optimizer lies in $[27.46,49.03]$ MW, strictly inside $[10,110]$; thus this certificate isolates the delayed smooth-composite theorem and is not presented as an active-box test. We take $R=5\mc V(0)$, rather than the zero-margin choice $R=\mc V(0)$, to certify a nontrivial initialization neighborhood.} Table~\ref{tab_certificate} reports every design and derived constant. In particular,
	$\wt\beta_1=3.8125\times10^{-3}$, $\wt\beta_2=2.74773\times10^{-1}$, and $\wt\beta_3=2.4447\times10^{-2}$ are positive, while
	$\varepsilon_\rho=1.8355\times10^{-4}<\kappa_VR=8.2643\times10^{-3}$ closes the invariant-sublevel check. {The proof auxiliaries were retuned by a deterministic bounded search that minimized $\varepsilon_\rho/\kappa_V$ subject to positive descent coefficients and self-consistency. These auxiliaries enter only the analysis: changing them changes the certified bound, not Algorithm~\ref{alg_1} or its numerical trajectory. Among near-minimizers, we selected the rounded interior point in Table~\ref{tab_certificate} while retaining a finite-$N$ bound below $3.1$ at $N=6000$. This avoids replacing an asymptotically weak certificate with a finite-horizon one. The resulting certified steady-state Lyapunov level is
		\[
		\varepsilon_\rho/\kappa_V=7.3555<\mc V(0)=66.2355,
		\]
		whereas the untuned auxiliary values gave $198.18$. The theorem now guarantees an improvement from the reported initialization, not merely non-escape from $\Omega_R$.}
	
	{In Fig.~\ref{fig_delay_certified}, at state index $k=6000$, the cost gaps are $1.6142\times10^{-2}$ for $\mu=0.09$ and $3.3180\times10^{-2}$ for $\mu=0$, providing a within-region momentum comparison. The center panel plots the theorem's actual finite-$N$ ergodic bound, not a cost-gap surrogate. At horizon $N=6000$, the bound is $2.9437$, compared with an observed running residual of $1.0572$. Equality feasibility is preserved to floating-point precision. For the auxiliary values in Table~\ref{tab_certificate}, condition~\eqref{eq_beta2} becomes $\widetilde\beta_2=2.1-19.5\mu-8.67\mu^2>0$, so the actual ceiling from this coefficient is $\mu<0.10298$; the reported $\mu=0.09$ is about $13\%$ below that ceiling.}
	
	{For the fixed Table~\ref{tab_certificate} auxiliaries and $R=5\mc V(0)$, solving $\varepsilon_\rho(\rho)=\kappa_VR$ gives the explicit certified ceiling $\rho_{\max}=6.7091\times10^{-4}$. Thus the theorem-certified resolution remains more than an order of magnitude finer than the empirical $\rho=1/64$ used in Section~\ref{sec_sim_quant}. Increasing $R$ cannot manufacture certification: since $G_R^2/R\to4u^2/v$, both $\varepsilon_\rho$ and $\kappa_VR$ scale linearly in $R$ to leading order, and the fixed-auxiliary ceiling saturates at $7.0528\times10^{-4}$ as $R\to\infty$. A separate deterministic search over the same auxiliary box, with $\mu=0$ and the same finite-$N$ guard, reaches only $\rho_{\max}=9.1920\times10^{-4}$. Thus even favorable retuning remains at the $10^{-3}$ scale, and the large gap to $1/64$ is structural rather than evidence of an unsearched certificate. Condition~\eqref{eq_self_consistency} is fundamentally a restriction on the quantization resolution and proof auxiliaries, not a consequence of choosing a loose initialization sublevel.}

	\begin{table*}[t]
		\centering
		\scriptsize
		\caption{{Complete numerical certificate for Theorem~\ref{thm_conv}.}}
		\begin{tabular}{|l|c|l|c|l|c|}
			\hline
			\multicolumn{6}{|c|}{\textit{Implemented problem and algorithm parameters}} \\ \hline
			$n$ & 25 & $\lambda_2=\lambda_n$ & 1 & $d_{\max}$ & 0.96 \\ \hline
			$v$ & 0.03 & $u$ & 0.55 & $\bar\tau$ & 2 \\ \hline
			$\rho$ & $10^{-4}$ & $\eta_\tau$ & 0.0095 & $\mu$ & 0.09 \\ \hline
			\multicolumn{6}{|c|}{\textit{Proof auxiliaries (analysis only)}} \\ \hline
			${\omega_w}$ & 4.2 & ${\omega_H}$ & 0.12 & $(\epsilon,\gamma,\xi)$ & $(39,2,0.125)$ \\ \hline
			\multicolumn{6}{|c|}{\textit{Derived certificate quantities}} \\ \hline
			$\psi^\star$ & 5.922131 & $\mc V(0)$ & 66.235485 & $R=5\mc V(0)$ & 331.177425 \\ \hline
			$G_R$ & 121.496737 & $C_H$ & 9.292800 & $C_{p,\rho}$ & 0.0141717 \\ \hline
			$C_{\times,\rho}$ & 0.00369055 & $C_{g,\rho}$ & 0.00354293 & $A_F$ & 152.777778 \\ \hline
			$\kappa_V$ & $2.49543\!\times\!10^{-5}$ & $\wt\beta_1$ & 0.00381247 & $\wt\beta_2$ & 0.274773 \\ \hline
			$\wt\beta_3$ & 0.0244471 & $\varepsilon_\rho$ & 0.000183553 & $\kappa_VR$ & 0.00826431 \\ \hline
			certificate & satisfied & $\mu=0$ baseline & satisfied & $\varepsilon_\rho/\min_i\wt\beta_i$ & 0.0481453 \\ \hline
			$\varepsilon_\rho/\kappa_V$ & 7.35554 & finite-$N$ bound ($N=6000$) & 2.94371 & $\rho_{\max}$ & 0.00067091 \\ \hline
		\end{tabular}
		\label{tab_certificate}
	\end{table*}

\begin{figure*}[t]
	\centering
	\includegraphics[width=0.99\textwidth]{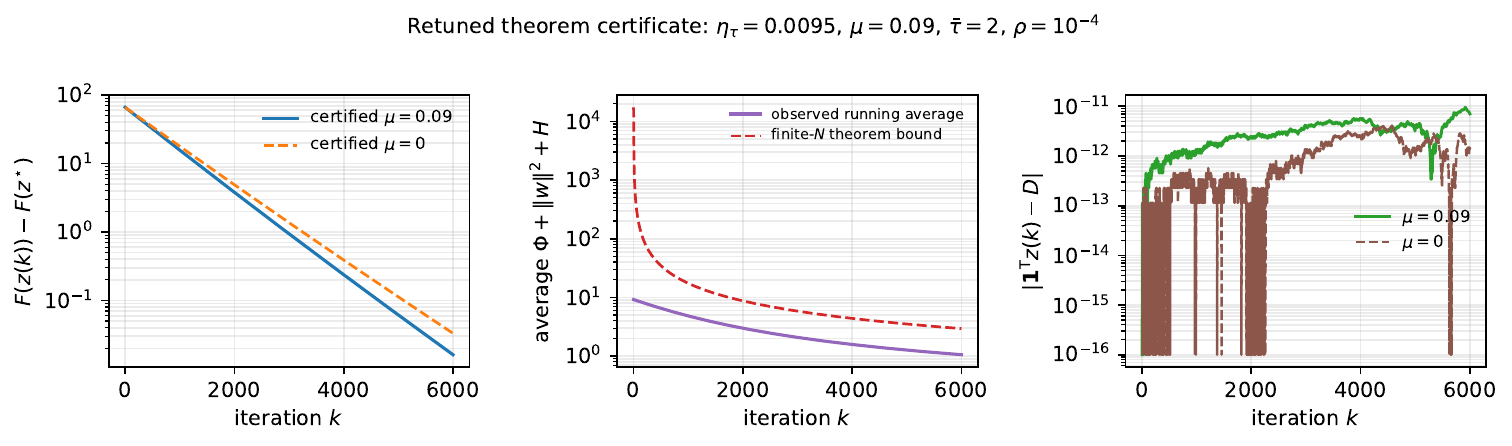}
	\caption{{Retuned theorem-certified delayed experiment. Left: composite-cost gaps for the certified $\mu=0.09$ dynamics and a separately verified $\mu=0$ baseline. Center: running average of $\Phi+\|w\|^2+H$ and its finite-$N$ bound $[\mc V(0)/N+\varepsilon_\rho]/\min_i\wt\beta_i$. Right: equality-feasibility errors; their gradual increase remains at floating-point scale and reflects accumulated roundoff in algebraically cancelling updates, not loss of the conservation invariant. Parameters and coefficients are in Table~\ref{tab_certificate}.}}
	\label{fig_delay_certified}
\end{figure*}

\begin{figure}[t]
	\centering
	\includegraphics[width=3.25in]{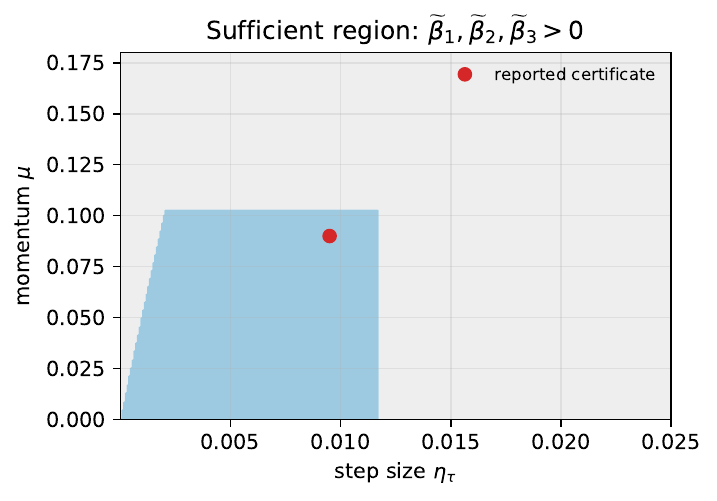}
	\caption{{Sufficient $(\eta_\tau,\mu)$ region for the fixed auxiliary parameters in Table~\ref{tab_certificate}; the red point is the reported certificate.}}
	\label{fig_admissible_region}
\end{figure}

\subsection{Delayed Stress Test Beyond Theorem Assumptions} \label{sec_sim_delay}

We now turn to a delayed stress test beyond the theorem assumptions.

We test with $n = 25$ nodes, ER graph with $20\%$ link probability, non-smooth penalty~\eqref{eq_log_penalty0} with $m_i = 10$, $M_i = 110$, $\sigma = 3$, $D = 1000$\,MW, $d_i = 40$, $\mu = 0.8$, $\eta_\tau = 0.15$, $\rho = 0.0312$, and symmetric delays with bounds $\bar\tau = 2, 5, 10, 16$.

	No theorem step-size bound is reported for this stress test. It uses the nonsmooth hinge penalty and $\mu=0.8>1/\sqrt3$, so Assumption~\ref{ass_cvx} and condition~\eqref{eq_beta2} fail regardless of the chosen $\eta_\tau$. The earlier partial check of only $\wt\beta_1$ and $\wt\beta_3$ was therefore invalid and has been removed.

\textit{Interpretation of the delayed experiments.} {For the representative $\bar\tau=10$ run, anytime equality feasibility (Lemma~\ref{lem_feas}) is illustrated by the constant weighted average in the top-right panel of Fig.~\ref{fig_delay}; the lemma establishes the conservation identity for every reported paired symmetric-delay bound.} {These curves are only an empirical robustness study beyond the theorem assumptions. Convergence for $\bar\tau\in\{2,5,10\}$ and divergence for $\bar\tau=16$ are observations, not evidence that the sufficient theorem boundary has been crossed by a certified amount.}

\begin{figure}[!t]
	\centering
	\includegraphics[width=2.62in]{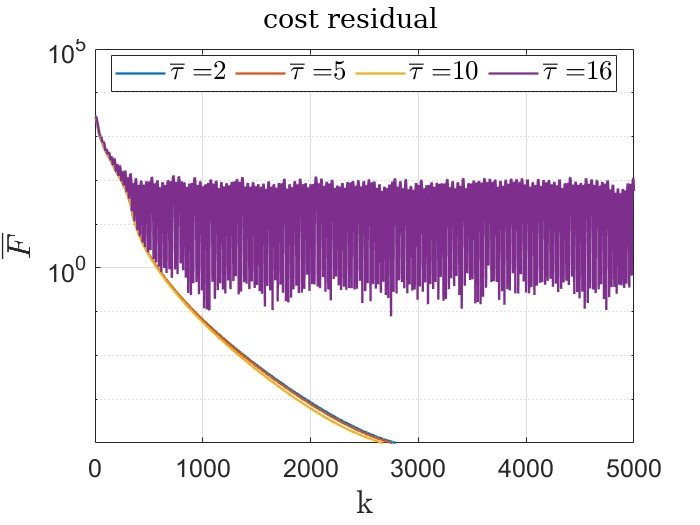}
	\includegraphics[width=2.62in]{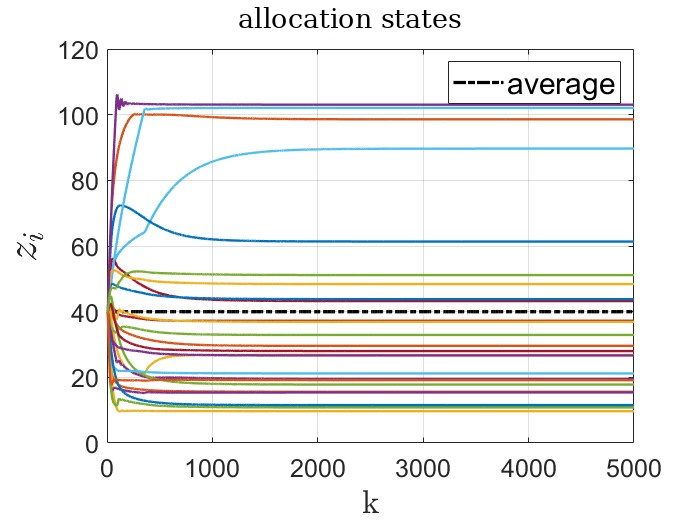}\\[-0.4ex]
	\includegraphics[width=2.62in]{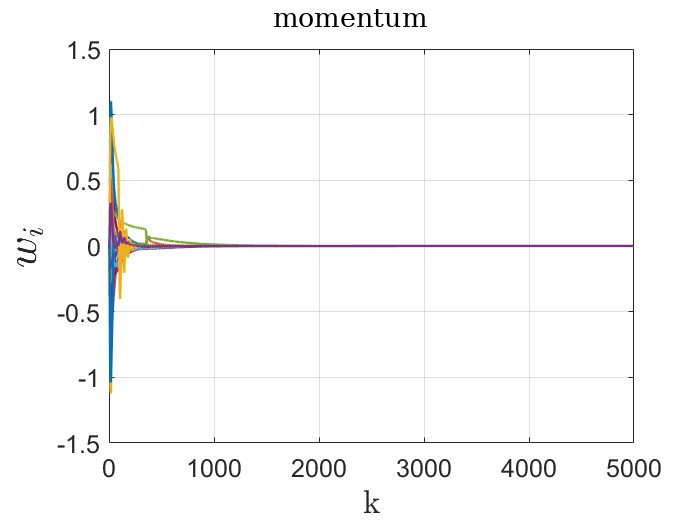}
	\caption{{Delayed stress test ($n=25$, $\eta_\tau=0.15$). Top-left: cost residuals for $\bar\tau\in\{2,5,10,16\}$. Top-right: allocation states for $\bar\tau=10$, with constant weighted average. Bottom: momentum variables for $\bar\tau=10$. The nonsmooth penalty and $\mu=0.8$ place these runs outside Theorem~\ref{thm_conv}.}}
	\label{fig_delay}
\end{figure}

\subsection{Nonquadratic Smooth Costs} \label{sec_sim_nonlinear}

To illustrate that the method is not limited to quadratic costs, we test smooth strongly convex original local objectives of the form
\[
f_i^0(z_i) = a_i z_i^2 + c_i \log(1 + e^{z_i}) + b_i z_i, \qquad a_i > 0,~ c_i \geq 0.
\]
The unpenalized Hessian satisfies
\[
2a_i \leq \partial_{z_i}^2 f_i^0(z_i) = 2a_i + c_i\frac{e^{z_i}}{(1+e^{z_i})^2} \leq 2a_i + \frac{c_i}{4},
\]
so $v=\min_i a_i=0.03$ and the unpenalized upper constant is at most $u_0=\max_i(a_i+c_i/8)=0.1025$. The dynamics, however, use the composite objective $\wh f_i=f_i^0+f_i^\sigma$. With $M_i-m_i=100$, $(\sigma,\kappa)=(2,2)$, and~\eqref{eq_penalty_u_sharp}, a valid global composite constant is
\[
u\leq u_0+\frac{\sigma\kappa}{8}+\frac{\sigma\kappa}{2}e^{-\kappa(M_i-m_i)/2}{<0.61}.
\]
We use $n = 40$ nodes with parameters drawn from the ranges $a_i \in [0.03, 0.04]$, $c_i \in [0, 0.5]$, $b_i \in [2, 4]$, the smooth penalty~\eqref{eq_log_penalty} with $m_i=10$, $M_i=110$, $\sigma = 2$, $\kappa = 2$, $D = 2000$\,MW, $\eta = 0.04$, $\rho = 0.0078$, and $\mu = 0.5$.

Fig.~\ref{fig_nonlinear} shows that the $\mu=0.5$ run reaches the common residual floor sooner than the $\mu=0$ run. {This is evidence only of an empirical transient benefit beyond the quadratic benchmark. }

\begin{figure}[htbp]
	\centering
	\includegraphics[width=3in]{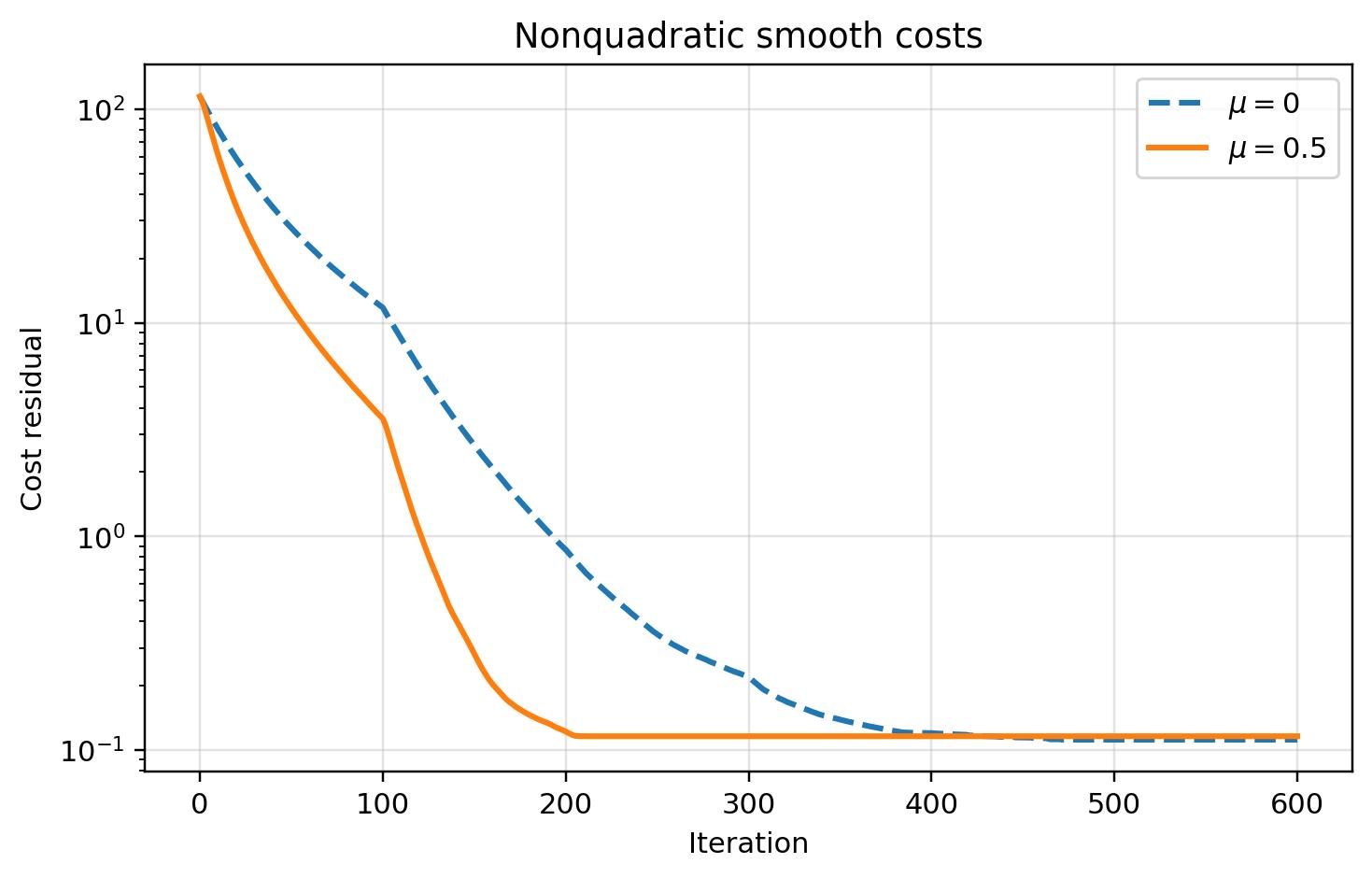}
	\caption{Cost residual for nonquadratic smooth strongly convex costs with $\mu = 0$ (dashed) and $\mu = 0.5$ (solid). In this empirical run, momentum improves the transient and both curves approach the same residual floor.}
	\label{fig_nonlinear}
\end{figure}

\subsection{Discussion and Limitations}

	Only Section~\ref{sec_sim_delay_certified} is claimed as a direct numerical certificate for Theorem~\ref{thm_conv}; it reports all constants, satisfies the invariant-sublevel test, and compares certified $\mu=0.09$ and $\mu=0$ cases. The other figures are empirical sensitivity, comparison, or stress studies. They illustrate equality conservation, transient acceleration, the signed generator-battery mechanism, and the effect of quantization, but their high momentum values, switching graphs, or nonsmooth penalties place them outside one or more sufficient assumptions. {The complete graph and very small $\rho=10^{-4}$ were chosen to keep the certificate transparent and tractable; this example does not demonstrate a sparse certified network or coarse quantization. Certification on sparse graphs is more conservative because $A_F$ grows as $\lambda_2$ decreases. The certificate itself uses generator nodes only, so $B=I$, but Section~\ref{sec_mixed_signed} now supplies a separate fixed sparse experiment with 12 generators and 8 batteries that directly exercises $B=\operatorname{diag}(+1,-1)$ and verifies signed equality conservation numerically. Extending the \emph{fully certified} benchmark to a mixed signed sparse instance remains future work. The theorem therefore supplies a guaranteed but conservative robust region: the empirical studies tolerate coarser quantization, larger momentum, larger step sizes, and switching or sparse graphs, but those observations are not promoted to guarantees.} This separation avoids using successful simulations as a substitute for the theorem conditions.

\section{Concluding Remarks and Future Directions} \label{sec_conc}

	This paper establishes exact per-iteration conservation of the coupling equality for a distributed signed-gradient method operating jointly with paired delays, fixed logarithmic quantization, and heavy-ball momentum. The associated discrete-time quantizer-plus-delay recursion is a hybrid system whose continuous states are driven by discrete quantizer-cell and delay selections. Its self-contained Lyapunov\nobreakdash-Krasovskii practical\nobreakdash-stability theorem applies on a static connected graph under bounded symmetric delays. The proof explicitly separates smooth delayed-gradient variation from discontinuous quantization error, treating the latter as an additive disturbance inside the Krasovskii functional, and uses an invariant-sublevel certificate rather than an assumed compact orbit. Equality feasibility is preserved at every iteration; finite smooth penalties, however, do not guarantee per-iterate box feasibility. Quantized equilibria within the certified sublevel have an explicit $O(\rho^2)$ gap relative to the smooth penalized optimizer, and exact convergence to that optimizer is recovered for $\rho=0$. Proposition~\ref{prop_spectral_hb} gives an exact momentum-acceleration characterization for the quadratic baseline; the fixed sparse-graph experiment confirms its modal recurrence and asymptotic rate, and the mixed generator-battery experiment directly exercises the signed dynamics. These conservation and discontinuous\nobreakdash-disturbance arguments are the main reusable contributions beyond the particular DERMS instance.

Several directions remain for future work:

\textit{Optimal momentum selection.} The admissible $(\eta_\tau, \mu)$ region defined by conditions~\eqref{eq_beta1}-\eqref{eq_beta3} is sufficient but likely conservative. {The quadratic delay-free optimum is given by~\eqref{eq_hb_opt}; identifying $\mu^*(\eta_\tau,\bar\tau,\rho)$ for nonlinear delayed quantized dynamics remains open.}

\textit{Asymmetric delays.} Anytime feasibility is proved under symmetric delays ($\tau_{ij} = \tau_{ji}$). Extending to fully asymmetric delays via pairwise exchange buffers or conservation mechanisms is an important open problem.

\textit{Uniform connectivity.} {The practical-stability proof currently requires a static connected graph. A windowed Lyapunov argument and a delay model compatible with changing edges would be needed for $B$-connectivity.}

\textit{Scalability and nonlinear costs.} Systematic experiments on large-scale networks and nonlinear cost models would further validate practical applicability.

\textit{Directed graphs.} The feasibility proof relies on undirected communication. Extending via push-sum or push-pull mechanisms \cite{nedic2015distributed,pu2021distributed} while preserving anytime feasibility remains open.

\bibliographystyle{elsarticle-num} 
\bibliography{bibliography}

\end{document}